\documentclass{article}

\usepackage{amssymb, amsmath, amsthm} % blackboard math symbols
\usepackage{booktabs}       % professional-quality tables
\usepackage{multirow}       % Table merge rows together

\usepackage{graphicx}

\usepackage{mathtools}

\usepackage{cite}

\usepackage{xcolor}

\usepackage{arxiv}

\usepackage{appendix}
\usepackage[ruled,vlined]{algorithm2e}

\usepackage[utf8]{inputenc} % allow utf-8 input
\usepackage[T1]{fontenc}    % use 8-bit T1 fonts
\usepackage{hyperref}       % hyperlinks
\usepackage{url}            % simple URL typesetting
\usepackage{booktabs}       % professional-quality tables
\usepackage{amsfonts}       % blackboard math symbols
\usepackage{nicefrac}       % compact symbols for 1/2, etc.
\usepackage{microtype}      % microtypography

\allowdisplaybreaks

\DeclareMathOperator*{\argmax}{argmax}

\DeclareMathOperator*{\RandMech}{\mathcal{A}}
\DeclareMathOperator*{\aux}{aux}

\DeclarePairedDelimiter\abs{\lvert}{\rvert}

\newcommand{\enc}{\text{\textsf{Enc}}}
\newcommand{\dec}{\text{\textsf{Dec}}}
\newcommand{\pk}{\text{\textsf{pk}}}
\newcommand{\sk}{\text{\textsf{sk}}}

\newcommand{\ZZ}{\mathbb{Z}}

\newcommand{\threshold}{t}
\newcommand{\uppervalue}{a}
\newcommand{\lowervalue}{b}

\newcommand{\onehotencoding}{z}

\newcommand{\ampNor}{\gamma}
\newcommand{\ratioSucTeach}{\tau}
\newcommand{\varinI}{v}

\newcommand{\normalInt}{L}

\newcommand{\density}{f}
\newcommand{\cumulDensity}{F}

\newcommand{\dbdist}{a}

\newcommand{\Idist}{\beta}

\newcommand{\nonNullThres}{b}

\newcommand{\amplitude}{A}
\newcommand{\initModulus}{b_i}
\newcommand{\bootModulus}{b_\theta}

\newtheorem{definition}{Definition}
\newtheorem{proposition}{Proposition}
\newtheorem{theorem}{Theorem}
\newtheorem{lemma}{Lemma}

\newenvironment{proof_sketch}{\noindent{\textit{Sketch of proof.}}~}

\newcommand\numberthis{\addtocounter{equation}{1}\tag{\theequation}}

\newtheorem{innercustomgeneric}{\customgenericname}
\providecommand{\customgenericname}{}
\newcommand{\newcustomtheorem}[2]{
  \newenvironment{#1}[1]
  {
   \renewcommand\customgenericname{#2}
   \renewcommand\theinnercustomgeneric{##1}
   \innercustomgeneric
  }
  {\endinnercustomgeneric}
}

\newcustomtheorem{customthm}{Theorem}
\newcustomtheorem{customprop}{Proposition}
\newcustomtheorem{customdef}{Definition}

\title{SPEED: Secure, PrivatE, and Efficient Deep learning}

\author{
    Arnaud Grivet S\'{e}bert\thanks{Contact author: arnaud.grivetsebert@cea.fr}~~$^{1}$\quad
    Rafael Pinot$^{1,2}$\quad
    Martin Zuber$^{1}$ \\
    \textbf{C\'{e}dric Gouy-Pailler}$^{1}$\quad
    \textbf{Renaud Sirdey}$^{1}$ \\ \\
    $^1$Institut LIST, CEA, Université Paris-Saclay, F-91120, Palaiseau, France \\
    $^2$Université Paris-Dauphine, PSL Research University, CNRS, LAMSADE, F-75016, Paris, France
}

\begin{document}

\maketitle

\begin{abstract}
We introduce a deep learning framework able to deal with strong privacy constraints. Based on collaborative learning, differential privacy and homomorphic encryption, the proposed approach advances state-of-the-art of private deep learning against a wider range of threats, in particular the honest-but-curious server assumption. We address threats from both the aggregation server, the global model and potentially colluding data holders. Building upon distributed differential privacy and a homomorphic argmax operator, our method is specifically designed to maintain low communication loads and efficiency.
The proposed method is supported by carefully crafted theoretical results. We provide differential privacy guarantees from the point of view of any entity having access to the final model, including colluding data holders, as a function of the ratio of data holders who kept their noise secret.
This makes our method practical to real-life scenarios where data holders do not trust any third party to process their datasets nor the other data holders. Crucially the computational burden of the approach is maintained reasonable, and, to the best of our knowledge, our framework is the first one to be efficient enough to investigate deep learning applications while addressing such a large scope of threats.
To assess the practical usability of our framework, experiments have been carried out on image datasets in a classification context. We present numerical results that show that the learning procedure is both accurate and private.

\end{abstract}

\section{Introduction}
\paragraph{Application scenarios.} We consider $n$ hospitals, each of which owns a (personal) labelled database composed of medical records from its patients and a model (e.g. neural network) trained on this database to predict if a new patient is victim of a given disease, say cancer. The hospitals' goal is to collaborate in order to improve the early detection of cancer. Building a model from a larger dataset than the personal databases would lead to improved detection capabilities. Nevertheless, these medical databases are highly-sensitive and the information they contain about the patients cannot be disclosed \cite{gdpr}. In such a setting, the hospitals wish to collaboratively train a global model while preserving confidentiality of their records. To do so, the idea is to rely on an aggregating institution (e.g. the World Health Organisation). This would amount to creating a three-party architecture: hospitals, aggregating institution, global model. Note that in our example, and in many real-world settings, all the training data providers may be recipients of the global model, or the global model may even be totally public. Hence, the global model may be exposed to attacks like membership inference attacks~\cite{shokri2017membership} that could indicate with high accuracy the probability that one patient was present in a database. Also, given a set of instances, the risk of a model inversion attack~\cite{wu2016methodology} which tries to infer sensitive attributes on the instances from a supposedly non-sensitive (often white-box) access to the model, is to be seriously taken into account as it would allow to infer for example that some of the hospital databases contain more ill patients than others. Besides, the aggregating institution might be the target of cyberattacks aimed at stealing data from it. For all these reasons, the three-party architecture we consider has to be resistant to threats coming from \emph{both the aggregation server and the global model recipients}.

Another motivating example, from the field of cybersecurity, is when several actors each hold a database of cybersecurity incident signatures that have occurred on their customer networks. The actors would rely on a third-party server to train the global model. In this scenario, it is a great security issue if the global model suffers from an attack (e.g. if the model features can be inferred~\cite{tramer2016stealing,yan2018cache,wang2018stealing} with limited access to the model). In this case, this would clearly leak some information on the detection capabilities of the actors, giving a clear advantage to cyberattackers on the networks they supervise.

\paragraph{Deployment scenario and threat model.} To perform the aggregation in a private way, we work in the tripartite setting summarised in Figure~\ref{fig:speed_scheme} and formally detailed in Section~\ref{sec:speed}. The \emph{student} (who holds the global model, a.k.a. the \emph{student model}) is the owner of the homomorphic encryption scheme under which encrypted-domain computations will be performed by the \emph{aggregation server}. This means that the student generates and knows both the encryption and decryption keys $\pk$ and $\sk$. Then, when being submitted an unlabelled input, the data holders (a.k.a. the \emph{teachers}) noise the predictions from their personal models, encrypt them under $\pk$ and send these encryptions to the server. The server has the responsibility to homomorphically perform the aggregation in order to produce an encryption of the output (e.g. a label) which will be sent back to the student and used by the latter for learning, after due decryption. \emph{Homomorphic encryption} thus provides a countermeasure to confidentiality threats on the teachers' predictions from the aggregation server, while the noise introduced by the actor addresses, via \emph{differential privacy}, the issue of attacks against the student model. In this setting, we assume that the student model is public or at least available to all the actors of the protocol, namely the teachers, the aggregation server and, of course, the student. Our mechanism is differentially private in this context, and our guarantees still hold against a malicious teacher, who has the information of the noise she generated, or even against colluding teachers (see Section~\ref{sec:dp_analysis}). On the contrary, we do not address threats whereby the student and the aggregation server collude in the sense that the student does not share $\sk$ with the server (in which case they would both get access to the teachers' predictions). We do not consider either threats where the aggregation server behaves maliciously, e.g. to prevent the student model from effectively learning from the teachers, leading to more or less stealthy forms of denial-of-service, or to perform a chosen ciphertext attack via selected queries to the student model. This is the typical scenario in which homomorphic encryption intervenes and our setting thus covers the threat model whereby the aggregation server is assumed to operate properly but may perform computations on observed data to retrieve information. This threat model is commonly known as the \emph{honest-but-curious} model~\cite{ishai2003extending, bonawitz2016practical, graepel2012ml}.

\begin{figure}[ht]
    \centering
    \includegraphics[width=0.9\textwidth]{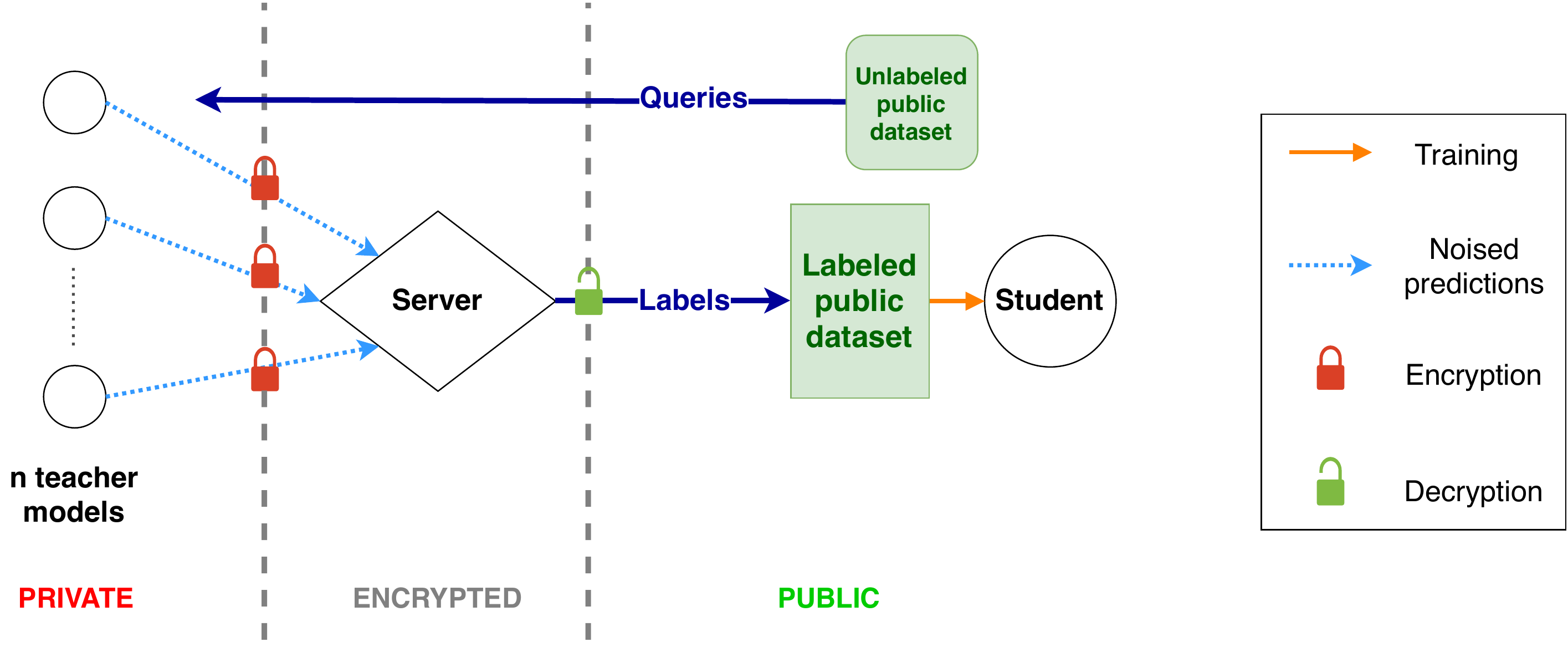}
    \caption{SPEED - Teacher models send to the aggregation server their encrypted noisy answers to the student's queries. The server homomorphically performs the aggregation in the encrypted domain and sends the result to the student model which decrypts it and uses it for training.}
    \label{fig:speed_scheme}
\end{figure}

\paragraph{Our contribution.} 
In this paper, we present a complete collaborative learning protocol which is secure along the whole workflow regarding a large scope of threats. We ensure protection of the data against any malicious actor of the protocol during the learning phase and prevent indirect information leakage from the final model using \emph{both} homomorphic encryption and differential privacy. While our framework is agnostic to the kind of models used by both the teachers and the student, to the best of our knowledge this is the first work with this level of protection to be efficient enough to apply to deep learning, therefore allowing very good accuracy on difficult tasks such as image classification, as shown by the experiments we ran. Our framework is also bandwidth-efficient and does not require more interactions than required by the baseline protocol.

\paragraph{Outline of the paper.} Section~\ref{sec:related_work} relates our work to the literature. In Section~\ref{sec:preliminaries}, we give some technical background on differential privacy and homomorphic encryption. We describe our SPEED framework in Section~\ref{sec:speed} and analyse its differential privacy guarantees in Section~\ref{sec:dp_analysis}. Section~\ref{sec:expe} presents our experimental results - SPEED achieves state-or-the-art accuracy and privacy with a mild computational overhead w.r.t previous works. Section~\ref{sec:conclusion} concludes the paper and states some open questions for further works.

\section{Related work}
\label{sec:related_work}
\paragraph{Differential privacy (DP).} Recent works considered to use differential privacy in collaborative settings close to the one we consider ~\cite{DBLP:journals/corr/abs-1812-01484,bhowmick2018protection,geyer2017differentially,chase2017private,papernot2016semisupervised,papernot2018scalable}. Among them, the most efficient technique in terms of accuracy and privacy guarantees is Private Aggregation of Teacher Ensembles (PATE) first presented in~\cite{papernot2016semisupervised} and refined in~\cite{papernot2018scalable}. PATE uses semi-supervised learning to transfer to the student model the knowledge of the ensemble of teachers by using a differentially private aggregation method. This approach considers a setting very close to ours with the notable difference that the aggregation server is trusted. Hence, applying PATE in our scenario makes the teacher models vulnerable. To tackle this issue, our work builds upon PATE idea with two key differences: we let the responsibility of generating the noise to the teachers and we add a layer of homomorphic encryption in order for the overall learning to be kept private. Another difference can also be noted. To derive privacy guarantees, PATE assumes that two databases $d$ and $d'$ are adjacent if only one sample of the personal database $d_i$ of one teacher $i$ changes, with the hypothesis that the personal databases $d_i$ are disjoint. We do not need this hypothesis and we only consider the teacher models, not the personal databases they use to train them. This leads us to a more powerful definition of adjacency: two databases $d$ and $d'$ are adjacent if they differ by one teacher.

\paragraph{Homomorphic Encryption (HE).} HE allows to perform computations over encrypted data. In particular, this can be used so that the model can perform both training and prediction without handling cleartext data. In terms of learning, the naive approach would be to have the training sets homomorphically encrypted, sent to a server for training to be done in the encrypted domain and the resulting (encrypted) model sent back to the participants for decryption. However, putting aside many subtleties, even by deploying all the arsenal available in the HE practitioner toolbox (batching, transciphering, etc.) this would be impractical as “classical” learning is both computation and know-how intensive and HE operations are intrinsically costly. As a consequence, there are only very few works that capitalise on HE for private training~\cite{graepel2012ml, hesamifard2017cryptodl, lou2020glyph} and inference~\cite{gilad2016cryptonets, juvekar2018gazelle} of machine learning tasks. Moreover, since some attacks can be performed in a black-box setting, the system is still vulnerable to attacks from the end user who has access to the decryption key. In our framework, we do not use HE directly to build the model, we use it as a mean for the aggregation to be kept private. That way, we are protected against potential threats from the aggregation server, which does not have the decryption key, and we keep a manageable computational overhead.

\paragraph{Federated learning.} Federated learning approaches gather several users who own data and make them collaborate in an iterative workflow in order to train a global model. The most famous federated learning algorithm is federated averaging\cite{mcmahan2016federated} which is a parallelised stochastic gradient descent. In a context of sensitive user data, several works proposed privacy-preserving federated learning or closely related distributed learning that make use of differential privacy~\cite{geyer2017differentially, shokri2015privacy}, cryptographic primitives~\cite{bonawitz2016practical, bonawitz2017practical, ryffel2020ariann} or both~\cite{chase2017private, sabater2020distributed, ryffel2018generic}. These methods require online communication between the parties whereas our solution takes advantage of homomorphic encryption and the existence of personal trained models to avoid online communication and drastically limit the interactions, that are both bandwidth-consuming and vulnerable to attacks.

\paragraph{Private aggregation.} Several approaches have been considered to limit the need for a trusted server when applying differential privacy, for example by considering local differential privacy~\cite{kasiviswanathan2011can, duchi2013local, kairouz2016extremal}. In practice it often results in applying too much noise, and maintaining utility can be difficult~\cite{ullman2018tight, kasiviswanathan2011can} especially for deep learning applications. In order to recover more accuracy while keeping privacy, some works combined decentralised noise distribution (\emph{a.k.a.} distributed differential privacy \cite{shi2011privacy}) and encryption schemes~\cite{rastogi2010differentially, acs2011have, goryczka2015comprehensive, shi2011privacy} in the context of aggregation of distributed time-series. Our work contributes to this line of research. However, our framework is the first one to be efficient enough to investigate deep learning applications while combining distributed DP and HE.
Another advantage of our solution concerns fault tolerance regarding the added noise. Some works addressed the problem of fault tolerance by making the server generate the noise that some users did not generate~\cite{bao2015new} while other works assume that the users themselves adapt the noise they generate to the possible failures~\cite{chan2012privacy}. In our setting, because of the encryption and the absence of communication between the teachers, we cannot suppose that any honest entity knows if some failures occurred. Moreover, the addition of noise to compensate a failure does not solve the problem of colluding teachers who may still send noise but do not keep it secret. In our protocol, the task of an honest actor (teacher or server) does not depend on the number of failures and we provide privacy guarantees as a function of the number of failures (see Section~\ref{sec:dp_analysis}) - it then suffices to assume an upper bound on this number to ensure a privacy guarantee.

\paragraph{Secure Multi-Party Computation (SMPC).} Secure Multi-Party Computation is a general approach that enables several parties to collaboratively perform a given computation without revealing to the other parties any more information than the result of this computation. In particular, secure aggregation regroups approaches which use SMPC techniques as one-time pads masking~\cite{bonawitz2016practical, bonawitz2017practical} or secret-sharing~\cite{danezis2013smart} to perform aggregation over sensitive data. Although these approaches are very close in intent to FHE-based ones, as the present one, they achieve different trade-offs. In a nutshell, when FHE is computation-intensive and non-interactive, SMPC puts more stress on protocol interactions. SMPC requires a lot of communication (garbled circuit generation and evaluation, oblivious input key retrieval, secret key sharing), both time-consuming and vulnerable to attacks, and needs in general that \emph{all} teachers play their role in the protocol for it to terminate - or fixing the fault tolerance issue implies additional rounds of communication~\cite{bonawitz2016practical, bonawitz2017practical}.
On the contrary, the FHE approach is more versatile, requires no interaction among the teachers and is robust to temporary teacher unavailability. Still, at the time of writing, it is the authors' opinion that both approaches are worth investigating in their own right (and this paper obviously belongs to the FHE thread of research).

\section{Preliminaries}
\label{sec:preliminaries}
\subsection{Differential privacy}

Differential privacy~\cite{dwork2006our} is a gold standard concept in privacy preserving data analysis. It provides a guarantee that under a reasonable privacy budget $(\epsilon,\delta)$, two adjacent databases produce statistically indistinguishable results. In this section, two databases $d$ and  $d'$ are said adjacent if they differ by at most one example.
\begin{definition}
   A randomised mechanism $\RandMech$ with output range $\mathcal{R}$ satisfies \emph{$(\epsilon,\delta)$-differential privacy} if for any two adjacent databases $d,d'$ and for any subset of outputs $S\subset \mathcal{R}$ one has \begin{equation*}
       \mathbb{P}\left[ \RandMech(d) \in S\right] \leq e^\epsilon \mathbb{P}\left[\RandMech(d') \in S\right] +\delta.
   \end{equation*}
\end{definition}

Let us also present a famous and widely used differentially private mechanism, known as the \emph{report noisy max} mechanism.

\begin{definition}
    Let $K\in \mathbb{N}^*$, and let $\mathcal{X}$ be a set that can be partitioned into $K$ subsets $\mathcal{X}_1$, \dots, $\mathcal{X}_K$. The mechanism that, given a database $d$ of elements of $\mathcal{X}$, reports $\argmax_{k\in [K]}\left[n_k+Y_k\right],$
where $[K]:=\{1,\dots,K\}$, $n_k:=|d \cap \mathcal{X}_k|$ and $Y_k$ is a Laplace noise with mean $0$ and scale $\frac1\gamma$, $\gamma\in \mathbb{R_+^*}$, is called \emph{report noisy max}.
\end{definition}

\begin{theorem}[\cite{dwork2014algorithmic}]
\label{th:dp_report_noisy_max}
Let $\RandMech$ be the report noisy max as above. Then $\RandMech$ is $(2\gamma,0)$-differentially private.
\end{theorem}

We now define the notion of \emph{infinite divisibility} that we will use to implement distributed differential privacy.
\begin{definition}
    A random variable $Y$ is said to be \emph{infinitely divisible} if, for any $m\in \mathbb{N}^*$, we can find a family $(X_{m, i})_{i\in [m]}$ of independent and identically distributed (i.i.d.) random variables such that $Y$ has the same distribution as $\sum_{i=1}^m X_{m,i}$.
\end{definition}

The following proposition from~\cite{kotz2001laplace} claims that the Laplace distribution is infinitely divisible~\footnote{Another well-known example of infinitely divisible probability distribution is the Gaussian distribution which can be seen as the sum of Gaussian distributions of well chosen scale parameter. In a possible further work, we could indeed replace the (distributed) Laplace noise by a (distributed) Gaussian noise.}, enabling to distribute its generation among an arbitrary number of agents.
\begin{proposition}[\cite{kotz2001laplace}]
    \label{prop:infinite_divisibility}
    Let $m\in \mathbb{N}^*$ and $\gamma \in \mathbb{R_+^*}$. Let $G_p^{(i)}$, for $(i, p) \in  [m] \times [2]$, be i.i.d. random variables following the Gamma distribution of shape $\frac{1}{m}$ and scale $\frac{1}{\gamma}$. Then $\sum_{i=1}^m\left(G_1^{(i)} - G_2^{(i)}\right)$ follows the Laplace distribution of mean $0$ and scale $\frac{1}{\gamma}$. The Laplace distribution is said to be \emph{infinitely divisible}.
\end{proposition}

\begin{definition}
Let $\RandMech$ be a randomised mechanism with output range $\mathcal{R}$ and $d$, $d'$ a pair of adjacent databases. Let $\aux$ denote an auxiliary input. For any $o\in \mathcal{R}$, the \emph{privacy loss} at $o$ is defined as 
\begin{equation*}
    c(o;\RandMech,\aux,d,d'):= \log \left( \frac{\mathbb{P}[\RandMech(\aux,d) =o]}{\mathbb{P}[\RandMech(\aux,d') =o]} \right).
\end{equation*}

We define the \emph{privacy loss random variable} $C(\RandMech,\aux,d,d')$ as
\begin{equation*}
    C(\RandMech,\aux,d,d') := c(\RandMech(d);\RandMech,\aux,d,d')
\end{equation*}
\emph{i.e.} the random variable defined by evaluating the privacy loss at an outcome sampled from $\RandMech(d)$.
\end{definition}

In order to determine the privacy loss of our protocol, we use a traditional two-fold approach. First of all, we determine the privacy loss per query and, in a second step, we compose the privacy losses of each query to get the overall loss. The classical composition theorem (see e.g.~\cite{dwork2014algorithmic}) states that the guarantees $\epsilon$ of sequential queries add up. Nevertheless, training a deep neural network, even with a collaborative framework as presented in this paper, requires a large amount of calls to the databases, precluding the use of this classical composition. Therefore, to obtain reasonable DP guarantees, we need to keep track of the privacy loss with a more refined tool, namely the moments accountant~\cite{abadi2016deep} that we introduce here, deferring the details of the method in Section A.1 of the appendix.

\begin{definition}
With the same notations as above, the \emph{moments accountant} is defined for any $l\in \mathbb{R_+^*}$ as
\begin{equation*}
    \alpha_{\RandMech}(l):=\max_{\aux,d,d'}\alpha_{\RandMech}(l;\aux,d,d') 
\end{equation*}
where the maximum is taken over any auxiliary input $\aux$ and any pair of adjacent databases $(d,d')$ and $\alpha_{\RandMech}(l;\aux,d,d'):=\log \left(\mathbb{E}\left[\exp(l C(\RandMech,\aux,d,d'))\right]\right)$ is the moment generating function of the privacy loss random variable.
\end{definition}

\subsection{Homomorphic encryption}
\label{sec:prelim_he}

Let us consider $\Lambda$ and $\Omega$ which respectively are the set of cleartexts (\emph{a.k.a.} the clear domain) and the set of ciphertexts (\emph{a.k.a.} the encrypted domain). A homomorphic encryption system first consists in two algorithms $\enc_\pk:\Lambda\longrightarrow\Omega$ and $\dec_\sk:\Omega\longrightarrow\Lambda$  where $\pk$ and $\sk$ are data structures which represent the public encryption key and the private decryption key of the cryptosystem. 

Homomorphic encryption systems are by necessity probabilistic, meaning that some randomness has to be involved in the $\enc$ function and that the ciphertexts set $\Omega$ is significantly much bigger than the cleartexts set $\Lambda$. Any (decent) homomorphic encryption scheme possesses the semantic security property meaning that, given $\enc(m)$ and polynomially many pairs $(m_i,\enc(m_i))$ it is hard\footnote{``Hard'' means that it requires solving a reference (conjectured) computationally hard problem on which the security of the cryptosystem hence depends. From a practical viewpoint, given a security target $\lambda$, the concrete parameters of a homomorphic scheme are chosen such that the best known (exponential-time) algorithms for solving the underlying reference problem require an order of magnitude of $2^\lambda$ nontrivial operations.} to gain any information on $m$ with a significant advantage over guessing. Most importantly, a homomorphic encryption scheme offers two other operators $\oplus$ and $\otimes$ where
\begin{itemize}
    \item $\enc(m_1)\oplus\enc(m_2)=\enc(m_1+m_2)\in\Omega$
    \item $\enc(m_1)\otimes\enc(m_2)=\enc(m_1m_2)\in\Omega$.
\end{itemize}

When these two operators are supported without restriction by a homomorphic scheme, it is said to be a Fully Homomorphic Encryption (FHE) scheme. A FHE with $\Lambda=\ZZ_2$ is Turing-complete and, as such, is \emph{in principle} sufficient to perform any computation in the encrypted domain with a computational overhead depending on the security target\footnote{Polynomial in $\lambda$.}. \emph{In practice}, though, the $\oplus$ and $\otimes$ are much more computationally costly than their clear domain counterparts which has led to the development of several approaches to HE schemes design each with their pros and cons. 

\paragraph{Somewhat HE (SHE).} Somewhat homomorphic encryption schemes, such as BGV \cite{BGV12} or BFV \cite{fan2012somewhat}, provide both operators but with several constraints. Indeed, in these cryptosystems the $\otimes$ operator is much more costly than the $\oplus$ operator and the cost of the former strongly depends on the \emph{multiplicative depth} of the calculation, that is the maximum number of multiplications that have to be chained (although this depth can be optimised~\cite{AubryCS19}). Interestingly, most SHE schemes offer a \emph{batching} capability by which multiple cleartexts can be packed in one ciphertext resulting in (quite massively) parallel homomorphic operations i.e.,
\begin{align}
\label{eq:batch}
\enc(m_1,...,m_{\kappa})\oplus\enc(m'_1,...,m'_{\kappa})=\enc(m_1+m'_1,...,m_{\kappa}+m'_{\kappa})
\end{align}
(and similarly so for $\otimes$). Typically, several hundreds such slots are available which often allows to significantly speed up encrypted-domain calculations.

\paragraph{Fully HE (FHE).} Fully homomorphic encryption schemes offer both the $\oplus$ and $\otimes$ operators without restrictions on multiplicative depth. At the time of writing, only the FHE-over-the-torus approach, instantiated in the TFHE cryptosystem~\cite{ChillottiGGI16}, offers practical performances. In this cryptosystem, $\oplus$ and $\otimes$ have the same constant cost. On the downside, TFHE offers no batching capabilities. To get the best of all worlds, the TFHE scheme is often hybridised with SHE by means of operators allowing to homomorphically switch among several ciphertext formats~\cite{chimera, lou2020glyph} to perform each part of calculation with the most appropriate scheme (see e.g.~\cite{embed2019}).

\section{SPEED: Secure, Private, and Efficient Deep Learning}
\label{sec:speed}
\subsection{A distributed learning architecture}
Let us consider a set of $n$ owners (a.k.a. \emph{teachers}) each holding a personal sensitive model $f_i$.
We assume that we also have an unlabelled public database $D$. The goal is to label $D$ using the knowledge of the private (teacher) models to train a collaborative model (a.k.a. \emph{student model}) mapping an input space $\mathcal{X}$ to an output space $[K]=\{1,\dots,K\}$. To do so while keeping the process private, we follow the setting illustrated by Figure~\ref{fig:speed_scheme} relying on a (distrusted) aggregation server:

\begin{enumerate}
    \item For every sample $x$ of the public database $D$, the student sends $x$ to the aggregator requesting it to output label for $x$. The aggregator forwards this request to the $n$ teachers.
    \item Each teacher $i$ labels $x$ using its own private model $f_i$. Then each teacher adds noise to the label (see Section~\ref{sec:noisy_labels}) and encrypts the noisy label before sending it to the aggregation server.
    \item The aggregator performs a homomorphic aggregation of the noisy labels and returns the result to the student model, namely the most common answered label (see Section~\ref{sec:homomorphic_aggregation}).
    \item The student, who owns the decryption key, decrypts the aggregated label and is then able to use the labelled sample to train its model.
\end{enumerate}

Our framework addresses two kinds of threats using two complementary tools.
On one hand, differential privacy protects the sensitive data from attacks against the student model. Indeed, some model inversion attacks~\cite{wu2016methodology} might disclose the training data of the student model, and especially the labels of database $D$. But differential privacy ensures that the noise applied to the teachers' answers prevents the aggregated labels from leaking information about the sensitive models $f_i$~\footnote{Thanks to the DP guarantees, the labels of $D$ could actually be published as well.}.
On the other hand, the homomorphic encryption of the teachers' answers prevents the aggregator to learn anything about the sensitive data while enabling it to blindly compute the aggregation.

\subsection{Noise generation and threat models}
\label{sec:noisy_labels}
When requested to label a sample $x$, each owner $i$ uses its model $f_i$ to infer the label of $x$. In order for the aggregator to compute the most common label in the secret domain, the owner must send a one-hot encoding of the label. That is, rather than sending $f_i(x)$, the $i$-th teacher sends a $K$-dimensional vector, say $\onehotencoding^{(i)}$, whose $f_i(x)$-th coordinate is an encryption of $1$ while all the other coordinates are encryptions of $0$. To guarantee differential privacy (see Section~\ref{sec:dp_analysis} for the formal analysis), the owner adds to this one-hot encoding a noise drawn from $G_1^{(i)} - G_2^{(i)}$ where the $G_1^{(i)}$ and $G_2^{(i)}$ are $2n$ i.i.d. $K$-dimensional random variables following the Gamma distribution of shape $\frac{1}{n}$ and scale $\frac{1}{\gamma}$, where $\gamma\in \mathbb{R_+^*}$. Then, $i$ sends the (encrypted) noisy one-hot encoded vector whose $k$-th coordinate corresponds to $\onehotencoding^{(i)}_k + G_{k,1}^{(i)} - G_{k,2}^{(i)}$.

Assuming that the aggregator has access to the student model, distributing the responsibility of adding the noise among all the teachers instead of delegating this task to the aggregator (see paragraph on centralised noise below) is necessary to protect the data against an honest-but-curious aggregator. Indeed, such an aggregator could use the information of the noise it generated to break the differential privacy guarantees and, potentially, recover the sensitive data by model inversion on the student model. Note that such an attack does not break the honest-but-curious assumption since the aggregator still performs its task correctly.

\paragraph{Beyond the honest-but-curious model} In a model that would go beyond the honest-but-curious aggregator hypothesis, the capability for the aggregator to add its own noise is even more harmful for the privacy (and of course, the accuracy) than not using noise at all. Indeed it gives the aggregator much more freedom to attack. As an example, think about a malicious aggregator that wants to know a characteristic $\chi$ on a particular teacher, called its victim. Given a query, for all $k\in [K]$, we write $n_k:=|\{i:f_i(x)=k\}|$ and call it the number of \emph{votes} for class $k$. Let us suppose that, for a given query, changing the value of the victim's characteristic $\chi$ from $\chi_0$ to $\chi_1$ also changes the victim's vote from a class $k_0$ to a class $k_1$. Hence, by denoting $n_{k_0}=\nu_0$ and $n_{k_1}=\nu_1$ if $\chi = \chi_0$ we get $n_{k_0}=\nu_0-1$ and $n_{k_1}=\nu_1+1$ if $\chi = \chi_1$. Then, if the aggregator knows all the $n_k$ for $k\in [K] \setminus \{k_0, k_1\}$ and knows $\nu_0$ and $\nu_1$ (which are the classical hypotheses in differential privacy), it can add just as much noise as needed for the class $k_0$ to be the argmax if and only if $\chi=\chi_0$~\footnote{For example, add $\nu_0 - \frac{1}{2} - n_k$ to all the classes except $k_0$ and $k_1$, $\nu_0 - 1 - \nu_1$ to the class $k_1$ and nothing to the class $k_0$.}. The result from the homomorphic argmax would then leak the information about the value of the victim's characteristic $\chi$.

\paragraph{Centralised noise generation} In a context in which the student model is kept private and, especially, not available to the aggregator, we can consider a centralised way of generating the noise. If we do not trust the teachers to generate the noise, we can charge the aggregator to do it, since it will not be able to use the knowledge of the noise to attack the sensitive data via the student model. The aggregator only needs to generate a Laplace noise (in the clear domain), and homomorphically add it to the unnoisy encryption of $n_k$ it receives from the teachers.
The infinite divisibility of the Laplace distribution (Proposition~\ref{prop:infinite_divisibility}) shows that the resulting noise is the same as in the case presented above in which each teacher generates an individual noise drawn from the difference of two Gamma distributions.
The privacy cost of one request is simply the privacy cost of the \emph{report noisy max}, namely $2\gamma$ (Theorem~\ref{th:dp_report_noisy_max}).\\

In a nutshell, we can consider the following different threat models:
\begin{itemize}
    \item honest (H) : the aggregation server performs its tasks properly and do not try to retrieve information from the data it has access to
    \item honest-but-curious (HBC) : the aggregation server performs its tasks properly but it may compute the available data to get sensitive information
    \item beyond honest-but-curious (BHBC) : the aggregation server performs the aggregation correctly but cannot be trusted to properly generate the noise necessary to the DP guarantees
\end{itemize}

Table~\ref{table:threat_models} summarises against which kind of server our protocol is protected, depending on the access the server has to the student model and on the way the noise is generated. As already emphasised, we focus on the case where the student model is public and the noise is distributively generated by the teachers because it is the most general model among the realistic threat models and thus gives the better tradeoff between flexibility and security.

\begin{table}[!htb]
    \centering
    \caption{Robustness of our framework depending on the availability of the student model and the noise generation}
    \begin{tabular}{c|c|c}
        & Private model & Public model \\\hline
        Centralised noise & \textbf{HBC} & \textbf{H} \\\hline
        Distributed noise & \textbf{BHBC} & \textbf{BHBC}
    \end{tabular}
    \label{table:threat_models}
\end{table}

\subsection{Technical details on the homomorphic aggregation}
\label{sec:homomorphic_aggregation}
\paragraph{Summing the noisy counts} The aggregation server receives the $n$ encrypted noisy labels and sums them up in the secret domain. Due to the infinite divisibility of the Laplace distribution, the server obtains a $K$-dimensional vector whose $k$-th ($k\in [K]$) coordinate is an encryption of:
\begin{equation*}
    \sum_{i=1}^{n} \left(\onehotencoding^{(i)}_k + G_{k,1}^{(i)} - G_{k,2}^{(i)} \right) = n_k + Y_k
\end{equation*}
where $n_k:=|\{i:f_i(x)=k\}|$ and $Y_k$ is a Laplace noise with mean $0$ and scale $\frac1\gamma$.

So far, we have only needed homomorphic addition which is a good start. Then an argmax operator must be performed after the summation. However, \emph{efficiently} handling the highly nonlinear argmax function by means of FHE is much more challenging.

\paragraph{Computing the argmax.} Most prior work on secure argmax computations use some kind of interaction between a party that holds a sensitive vector of values and a party that wants to obtain the argmax over those values. The non-linearity of the argmax operator presents unique challenges that have mostly been handled by allowing the two interested parties to exchange information. This means increased communication costs and, in some cases, information leakage. This is with the exception of~\cite{embed2019}. They provide a fully non-interactive homomorphic argmax computation scheme based on the TFHE encryption. We implemented and parametrised their scheme to fit the specific training problems presented in Section~\ref{sec:expe}.
We present here the main idea behind this novel FHE argmax scheme. For more details, see the original paper. The TFHE encryption scheme provides a \textit{bootstrap} operation that can be applied on any scalar ciphertext. Its purpose is threefold: switch the encryption key; reduce the noise; apply a non-linear operation on the underlying plaintext value. This underlying operation can be seen as a function 
\begin{equation*}
    g_{\threshold, \uppervalue, \lowervalue}(x) = 
     \begin{cases}
       \uppervalue & \quad \text{if} \quad x > \threshold  \\
       \lowervalue & \quad \text{if} \quad x < \threshold . \\
     \end{cases}
\end{equation*} 
One notable application is that of a "sign" bootstrap: we can extract the sign of the input with the underlying function $g_{0, 1, 0}(x)$. The argmax computation in the ciphertext space is made as follows. For every $k,k'$, $k \neq k'$, we compare the values $n_k+Y_k$ and $n_{k'}+Y_{k'}$ with a subtraction ($n_k+Y_k - n_{k'} - Y_{k'}$) and application of a sign bootstrap operation. This yields $\theta_{k,k'}$, a variable with value 1 if $n_k+Y_k > n_{k'}+Y_{k'}$ and $0$ otherwise. Therefore the complexity will be quadratic in the number of classes.  For a given $k$ we can then obtain a boolean truth value ($0$ or $1$) for whether $n_k+Y_k$ is the maximum value. To this end, we compute
\begin{equation*}
        \Theta_k = \sum_{i \neq k} \theta_{k,i}.
\end{equation*}
$n_k$ is the max if and only if, for all $ i$ one has $  \theta_{k,i} = 1$ \emph{i.e.} $\Theta_{k} = K-1$. We can therefore apply another bootstrap operation with $g_{K-\frac{3}{2}, 1, 0}$. If $\Theta_{k} = K-1$, the boostrap will return an encryption of $1$, and return an encryption of $0$ otherwise. Once decrypted, the position of the only non-zero value is the argmax. Because the underlying function $g_{\threshold, \uppervalue, \lowervalue}$ is applied homomorphically, its output is inherently probabilistic. In the FHE scheme used, an error is inserted in all the ciphertexts at encryption time to ensure an appropriate level of security. This means that if two values are too close, then the sign bootstrap operation might return the wrong result over their difference. The exact impact of this approximation on the accuracy is evaluated in Section~\ref{sec:expe}.

\paragraph{Remark.} Another solution would be to send the noisy histogram $n_k+Y_k$ of the counts for each class $k$ to the student and let her process the argmax in the clear domain. This could indeed be performed with a plain-old additively-homomorphic cryptosystem such as Paillier or (additive-flavored) ElGamal, avoiding the machinery of the homomorphic argmax. Nevertheless, this approach was put aside because sending the whole histogram instead of the argmax would provide much worse DP guarantees.

\section{Differential privacy analysis}
In this section, we will give privacy guarantees considering that two databases $d$ and $d'$ are adjacent if they differ by one teacher i.e. there exists $i_0\in [n]$ such that $f_{i_0} \neq f'_{i_0}$ and, for all $i\in [n]\setminus \{i_0\}$, $f_i = f'_i$. This definition of adjacency is quite conservative and is strictly larger than the definition of adjacency from~\cite{papernot2016semisupervised} (indeed, in the assumption whereby the personal teacher databases $d_i$ are disjoint, changing one sample from a personal database changes at most one teacher).

\label{sec:dp_analysis}
\paragraph{Robustness against colluding teachers.} As we have decided not to trust the aggregation server to generate the noise necessary to the privacy guarantees, we may also assume that a subset of teachers might be malicious and collude by communicating their generated noise, which gives the same DP guarantees from the point of view of a colluding teacher as if they would have not generated any noise and, to this extent, our protocol, which addresses this issue, is fault tolerant. The following theorem quantifies the privacy cost of such failures.
\\

In the following, we call $\RandMech$ the aggregation mechanism that outputs the argmax of the noisy counts. $\RandMech(d,Q)$ is the output of $\RandMech$ for the database $d$ and the query $Q$. Let $\ampNor\in \mathbb{R}^*_+$ be the inverse scale parameter of the distributed noise. Considering the DP guarantees from the point of view of an entity $\mathcal{E}$, let $\ratioSucTeach\in (0, 1)$ be the ratio of the teachers whose noise is ignored by $\mathcal{E}$.

\begin{theorem}
    \label{th:dp_per_query}
    Let us define $I \colon v \in \mathbb{R}^*_+ \mapsto \int_{0}^{+\infty}\left(t+v\right)^{\ratioSucTeach-1}t^{\ratioSucTeach-1}e^{-2 t}dt$ and $g \colon t \in \mathbb{R} \mapsto \frac{\int_{\ampNor t}^{+\infty}e^{-v}I(v)dv}{\int_{\ampNor (t + 2)}^{+\infty}e^{-v}I(v)dv}$.
    
    Then, from $\mathcal{E}$'s point of view, $\RandMech$ is $(\epsilon, 0)$-differentially private, with
    \begin{align*}
            \epsilon = \log\left(1 + 2\frac{\int_0^{\ampNor}e^{-v}I(v)dv}{\int_{2\ampNor}^{+\infty}e^{-v}I(v)dv}\right).
    \end{align*}
    Moreover, if $\ratioSucTeach > \frac{1}{2}$, g is differentiable in $0$ and $\RandMech$ is $(\epsilon', 0)$-differentially private, with
    \begin{align*}
            \epsilon' = \min \left[\epsilon, \log\left(g(0) - g'(0)\right)\right]
    \end{align*}
    where $g'(0) = \ampNor \frac{\frac{\Gamma(\ratioSucTeach)^2}{2}e^{-2\ampNor}I(2\ampNor) - I(0)\int_{2\ampNor}^{+\infty}e^{-v}I(v)dv}{\left(\int_{2\ampNor}^{+\infty}e^{-v}I(v)dv\right)^2}$.
\end{theorem}

\begin{proof_sketch}
    Adapting the proof of the privacy cost of the report noisy max from~\cite{dwork2014algorithmic}, we first show that, if we can find a function $M$ of $\ampNor$ and $\ratioSucTeach$ such that, for any $t\in \mathbb{R}$, $g(t) \leq M$, then $\RandMech$ is $(\log(M), 0)$-differentially private. This motivates us to find an upper bound of $g$.
    
    To do so, we prove that $g$ has a maximum on $\mathbb{R}$ and that this maximum is reached on the interval $[-1; 0]$.
    On one hand, we show that, for all $t\in [-1; 0]$, $g(t) \leq 1 + 2\frac{\int_0^{\ampNor}e^{-v}I(v)dv}{\int_{2\ampNor}^{+\infty}e^{-v}I(v)dv}$.
    On the other hand, we prove that, if besides $\ratioSucTeach > \frac{1}{2}$, then $g$ is concave on $[\argmax(g); 0]$ and thus, for all $t\in [-1; 0]$, $g(t) \leq g(0) - g'(0)$ (note that $g$ is not differentiable in $0$ if $\ratioSucTeach \leq \frac{1}{2}$).\\
\end{proof_sketch}

Denoting $S$ the subset of teachers who are honest (i.e. do not collude), this theorem allows us to control the privacy cost by the ratio $\ratioSucTeach$ of the teachers who kept their noise secret, from the point of view of both:
\begin{itemize}
    \item a colluding teacher, taking $\ratioSucTeach = \frac{|S|}{n}$
    \item an honest teacher, taking $\ratioSucTeach = \frac{n-1}{n}$
    \item any entity who has access to the student model but is not a teacher, taking $\ratioSucTeach = 1$
\end{itemize}

Note that we can also use Theorem~\ref{th:dp_per_query} in the hypothesis whereby the colluding teachers publish their noise (to the whole world), adapting $\ratioSucTeach$ in consequence~\footnote{e.g. the privacy guarantee for an honest teacher would be computed with $\ratioSucTeach = \frac{|S|-1}{n}$.}. For $\ratioSucTeach=1$, the privacy guarantee is given by $\underset{\ratioSucTeach\to 1}{\lim} \epsilon'$ which, as shown by Proposition~\ref{prop:limit_tau}, is the classical bound of the \emph{report noisy max} with a centralised Laplace noise.

\begin{proposition}
    \label{prop:limit_tau}
    For all $\ampNor\in \mathbb{R}_+^*$, $\underset{\ratioSucTeach\to 1}{\lim}\left[\log(g(0) - g'(0))\right] = 2\ampNor$.
\end{proposition}

Furthermore, Proposition~\ref{prop:limit_gamma} shows that, naturally, the privacy cost tends to be null when the noise becomes infinitely large ($\ampNor$ approaches $0$).

\begin{proposition}
    \label{prop:limit_gamma}
    For all $\ratioSucTeach\in (0, 1)$, $\underset{\ampNor\to 0}{\lim}\left[\log\left(1 + 2\frac{\int_0^{\frac{\ampNor}{2}}e^{-v}I(v)dv}{\int_{\ampNor}^{+\infty}e^{-v}I(v)dv}\right)\right] = 0$.
\end{proposition}

Let us also give an upper bound of the probability that the noisy argmax is different from the true argmax.
\begin{proposition}
    \label{prop:maj_proba_false_argmax}
    Let $k^*$ be the class corresponding to the true argmax.
    
    If $\ratioSucTeach \in (\frac{1}{2}; 1)$,
    \begin{align*}
        \mathbb{P}[\RandMech(d;Q)\neq k^*] \leq \sum_{k\neq k^*} e^{-\ampNor \Delta_k} \left[\frac{1}{2} + \frac{(\ampNor \Delta_k)^{2\ratioSucTeach-1}}{\ratioSucTeach 2^{4\ratioSucTeach-2} \Gamma(\ratioSucTeach)^2}\right]
    \end{align*}
    where $\Delta_k \coloneqq n_{k^*} - n_k$ for any $k\in [K]$ and $\Gamma:\beta\in \mathbb{R}^*_+ \mapsto \int_0^{+\infty}t^{\beta-1}e^{-t}dt$ is the gamma function.
    
    If $\ratioSucTeach \in (0; \frac{1}{2}]$,
    \begin{align*}
        \mathbb{P}[\RandMech(d;Q)\neq k^*] \leq \sum_{k\neq k^*} e^{-\ampNor \Delta_k}\left[\frac{1}{2} + \frac{(\ampNor \Delta_k)^{\frac{\ratioSucTeach}{2}}}{\ratioSucTeach 2^{\frac{5}{2}\ratioSucTeach-1} \Gamma(\ratioSucTeach)^2} \times \left(\frac{3}{2}\ratioSucTeach\right)^{\frac{3}{2}\ratioSucTeach}\left(\frac{2}{\ratioSucTeach}-3\right)^{1-\frac{3}{2}\ratioSucTeach}\right].
    \end{align*}
\end{proposition}

\begin{proof_sketch}
    The event $(\RandMech(d;Q)\neq k^*)$ is the union of the events $(n_k+Y_k \geq n_{k^*} + Y_{k^*})$, for $k\in [K]\setminus \{k^*\}$, and thus $\mathbb{P}[\RandMech(d;Q)\neq k^*] \leq \sum_{k\neq k^*} \mathbb{P}(n_k+Y_k \geq n_{k^*} + Y_{k^*})$.
    We remark that, for any $k\in [K]\setminus \{k^*\}$,
    \begin{align*}
        &\mathbb{P}(n_k+Y_k \geq n_{k^*} + Y_{k^*}) = \mathbb{P}(Y_{k^*} \leq Y_k - \Delta_k)\\
        &\quad= \int_{-\infty}^0 \density(t) \cumulDensity(t-\Delta_k)dt + \int_0^{\Delta_k} \density(t) \cumulDensity(t-\Delta_k)dt + \int_{\Delta_k}^{+\infty} \density(t) \cumulDensity(t-\Delta_k)dt
    \end{align*}
    where $\density \colon u\in \mathbb{R}^* \mapsto \frac{\ampNor}{\Gamma(\ratioSucTeach)^2}e^{-\ampNor \abs{u}}I(\ampNor \abs{u})$ and $\cumulDensity \colon t\in \mathbb{R} \mapsto \int_{-\infty}^t \density(u)du$.
    
    We show that $\int_{\Delta_k}^{+\infty} \density(t) \cumulDensity(t-\Delta_k)dt \leq \frac{3}{8}e^{-\ampNor \Delta_k}$ and $\int_{-\infty}^0 \density(t) \cumulDensity(t-\Delta_k)dt \leq \frac{1}{8}e^{-\ampNor \Delta_k}$.
    Moreover, using Hölder's inequality, we show that, for all $q\in (\frac{1}{1-\ratioSucTeach}; +\infty)$, calling $p:=\frac{1}{1-\frac{1}{q}}$, $\int_0^{\Delta_k} \density(t) \cumulDensity(t-\Delta_k)dt \leq \frac{e^{-\ampNor \Delta_k}}{\ratioSucTeach 2^{4\ratioSucTeach-2+\frac{1}{q}} \Gamma(\ratioSucTeach)^2} \times \frac{(\ampNor \Delta_k)^{2\ratioSucTeach-1+\frac{1}{q}}}{p^{\frac{1}{p}}[q(1-\ratioSucTeach)-1]^{\frac{1}{q}}}$.
    For $\ratioSucTeach > \frac{1}{2}$, we take the particular (and classic) case of the limit of the previous bound when $q$ tends to $+\infty$.
    For $\ratioSucTeach \leq \frac{1}{2}$, we take $q = \frac{1}{1-\frac{3}{2}\ratioSucTeach}$.\\
\end{proof_sketch}

Theorem~\ref{th:dp_per_query} and Proposition~\ref{prop:maj_proba_false_argmax} serve as building blocks to which we apply the following theorem from~\cite{papernot2016semisupervised}.

\begin{theorem}[\cite{papernot2016semisupervised}]
    \label{th:AlgoMomentAccountance}
    Let $\epsilon, l\in \mathbb{R}_+^*$. Let $\RandMech$ be a $(\epsilon,0)$-differentially private mechanism and $q \geq \mathbb{P}[\RandMech(d)\neq k^*]$ for some outcome $k^*$. If $q < \frac{e^{\epsilon} -1 }{e^{2\epsilon} -1}$, then for any additional information $\aux$ and any pair $(d, d')$ of adjacent databases, $\RandMech$ satisfies
    \begin{align*}
        &\alpha_{\RandMech}(l;\aux,d,d') \leq \min\left[\epsilon l, \frac{\epsilon^2 l (l+1)}{2}, \log\left((1 -q)\left(\frac{1 -q}{1 -e^{\epsilon}q}\right)^l + qe^{\epsilon l}\right)\right].
    \end{align*}
\end{theorem}

As in~\cite{papernot2016semisupervised}, Theorem~\ref{th:AlgoMomentAccountance} coupled with some properties of the moments accountant (composability and tail bound) allows one to devise the overall privacy budget $(\epsilon, \delta) $ for the learning procedure (see Section~\ref{sec:expe} for numerical results). We refer the interested reader to Section A of the appendix for more details and for the extended proofs of our claims.

\paragraph{Influence of the cryptographic layer.}One must be aware that the cryptographic layer perturbates the noisy votes because the computation of the homomorphic argmax has a small probability of error. Although this topic deserves further investigations, we make the assumption that these perturbations are negligible and that they do not change the privacy guarantees as they basically constitute an additional noise on the votes. We further discuss this point in Appendix A.3.

\section{Experimental results}
\label{sec:expe}
The experiments presented below enable us to validate the accuracy of our framework on well-known image classification tasks and illustrate the practicality of our method in terms of performance, since the computational overhead due to the homomorphic layer remains reasonable.
The source codes necessary to run the following experiments are available on \url{https://github.com/Arnaud-GS/SPEED}.

\paragraph{HE time overhead.} We implemented the homomorphic argmax computation presented in Section \ref{sec:homomorphic_aggregation}. Without parallelizing, a single argmax query requires just under $4$ seconds to compute on an Intel Core i7-6600U CPU. Importantly, this does not depend on the input data. The costliest operation is the computation of $\theta$. Any other part of the scheme is negligible in comparison. Therefore, once the parameters are set, the time performance depends solely on the number of classes (the number of bootstrap comparisons is quadratic in the number of classes). As such, $100$ queries require $6.5$ minutes and $1000$ queries $65$ minutes. Of course, the queries can be performed in parallel to decrease the latency allowing for much more challenging applications.

\paragraph{Homomorphic argmax accuracy.} As we mention in Section~\ref{sec:homomorphic_aggregation}, the homomorphic computation of the argmax is inherently probabilistic. This is due both to the noise added to any ciphertext at encryption time, and to limitations of the bootstrapping operation in terms of accuracy. On MNIST dataset~\cite{lecun1998mnist}, we evaluate the method with $\ratioSucTeach=1/0.9/0.7$ and compare the cleartext argmax to our homomorphic argmax. Our implementation of the HE argmax has an average accuracy of $99.4\%$, meaning that it retrieves the cleartext argmax $99.4 \%$ of the time.
\\
To obtain a more general and conservative measure of the inherent accuracy  of the HE argmax (which can be applied on any dataset), we make the teachers give uniformly random answers to the queries. In this setting, most counts $n_k$ are likely to be close to one another, which makes even a classical argmax useless. This kind of scenario can be seen as worst-case, since the teacher voting is adversarial to argmax computation. Even in this scenario, and with the same parameters as for MNIST, our implementation of the HE argmax algorithm still produces an average accuracy of $90 \%$. Hence, an accuracy of $90 \%$ can be considered a lower bound for any adaptation of this argmax technique to other datasets. Yet in practice a tweaking of the parameters can yield a better accuracy even for this worst-case scenario, at the cost of time efficiency.

\paragraph{Learning setup.} To evaluate the performances of our framework, we test our method on MNIST~\cite{lecun1998mnist} and SVHN~\cite{netzer2011reading} datasets. To represent the data holders, we divide the training set in $250$ equally distributed and disjoint subsets, keeping the test set for learning and evaluation of the student model. Then we apply the following procedures. We refer the interested reader to Section C of the appendix for more details on the hyper-parameters and learning procedure.

\begin{itemize}
\item \emph{Teacher models.} For MNIST, given a dataset, a data holder builds a local model by stacking two convolutional layers with max pooling and a fully connected layer with ReLu activations. Two additional layers have been added for SVHN.
\item \emph{Student model.} Following the idea from~\cite{papernot2016semisupervised}, we train the student in a semi-supervised fashion. Unlabelled inputs are used to estimate a good prior distribution using a GAN-based technique first introduced in~\cite{salimans2016improved}. Then we use a limited amount of queries ($100$ for MNIST, $500$ for SVHN) to obtain labelled examples which we use to fine tune the model.
\end{itemize}

For MNIST experiments, as the student model can substantially vary based on the selected subset of labelled examples, the out-of-sample accuracy has been evaluated $15$ times, with $100$ labelled examples sampled from a set of $9000$ ones. For each experiment, the remaining $1000$ examples have been used to evaluate the student model accuracy.
For SVHN, the computations being much more heavy, the out-of-sample accuracy has been evaluated $3$ times, with $500$ examples sampled from a set of $10000$ ones. We used $16032$ examples to test the student model accuracy.

\paragraph{Performances on MNIST.} Table~\ref{table:mnist} displays our experimental results for SPEED with MNIST and compares them to a non-private baseline (without DP or HE) and to the framework that we call Trusted which assumes that the server is trusted and thus only involves DP and not HE.
Trusted can be considered as PATE framework from~\cite{papernot2016semisupervised} with some subtle differences: the noise is generated in a distributed way in Trusted and the notion of adjacency is larger. Even if the inverse noise scale $\gamma$ we use is greater than the one in~\cite{papernot2016semisupervised} ($0.1$ instead of $0.05$), which should lead to a worse DP guarantee, an argmax-specific analysis of the privacy cost per query allowed us to provide a better DP guarantee ($\epsilon = 1.41$ instead of $\epsilon = 2.04$ with $\delta = 10^{-5}$ and $100$ queries).
To be more conservative in terms of accuracy, the experiments were run considering that the colluding teachers did not generate any noise, which does not change anything in terms of DP. That is why, in spite of the variability of the accuracy, we observe a tradeoff between accuracy and DP. Indeed, even if the reported average accuracy does not vary much across conditions, consistent rankings of the methods have been observed, confirming the expected average rank of the method based on the amount of added noise. As expected, the best DP guarantee ($\epsilon=1.41$) is obtained when all the teachers generated noise ($\ratioSucTeach=1$), but this is the case where the accuracy is the lowest. On the contrary, when some teachers failed to generate noise ($\ratioSucTeach=0.9$ and $\ratioSucTeach=0.7$), the counts are more precise, leading to a slightly better accuracy but worse DP guarantees. It should also be noted that the variance is high in each condition. It masks the fact that the distribution is highly skewed, with a majority of results in the $97.5\%-98.5\%$ range, and a few samplings yielding an out-of-sample accuracy around $90\%$.

\begin{table}[!htb]
    \centering
    \caption{Results for MNIST dataset with 250 teachers and 100 student queries. We used an inverse noise scale $\gamma=0.1$. The DP guarantees, computed by composability with the moments accountant method over the 100 queries, are given for $\delta=10^{-5}$.}
    \begin{tabular}{lccc}
        \hline\noalign{\smallskip}
        Framework & $\epsilon$ & Acc. ($\pm$ std)  [\%]  & HE overhead \\
        \noalign{\smallskip}\hline\noalign{\smallskip}
        Non-private      & -      & $96.22$ ($\pm 2.27$) & - \\
        Trusted          & $1.41$ & $95.95$ ($\pm 2.97$) & - \\
        \noalign{\smallskip}\hline\noalign{\smallskip}
        $\ratioSucTeach=1$   & $1.41$ & $95.91$ ($\pm 2.57$) & \multirow{3}{*}{$6.5 \; $min} \\
        $\ratioSucTeach=0.9$ & $1.66$ & $96.02$ ($\pm 2.92$) &  \\
        $\ratioSucTeach=0.7$ & $2.37$ & $96.06$ ($\pm 2.61$) &  \\
        \noalign{\smallskip}\hline
    \end{tabular}
    \label{table:mnist}
\end{table}

Figure~\ref{fig:mnist_dp_gamma} shows the evolution of our DP guarantee as a function of $\ampNor$, with $\ratioSucTeach=0.9$ fixed. Note that the privacy cost decreases for $\ampNor \geq 2$ which may seem counterintuitive but the reason is thoroughly explained in Section A.4 of the appendix. Anyway, we observed empirically that the privacy cost has a finite limit in $+\infty$ (approximately $2.87$) and remains greater than this limit for any $\ampNor \geq 2$. The asymptote is shown by a dashed line on Figure~\ref{fig:mnist_dp_gamma}.

Figure~\ref{fig:mnist_dp_tau} shows the evolution of the DP guarantee as a function of $\ratioSucTeach$, with $\ampNor=0.1$ fixed. As explained before, the greater $\ratioSucTeach$, the better the DP guarantee.

\begin{figure}[!tbp]
  \centering
  \begin{minipage}[b]{0.45\textwidth}
    \includegraphics[width=\textwidth]{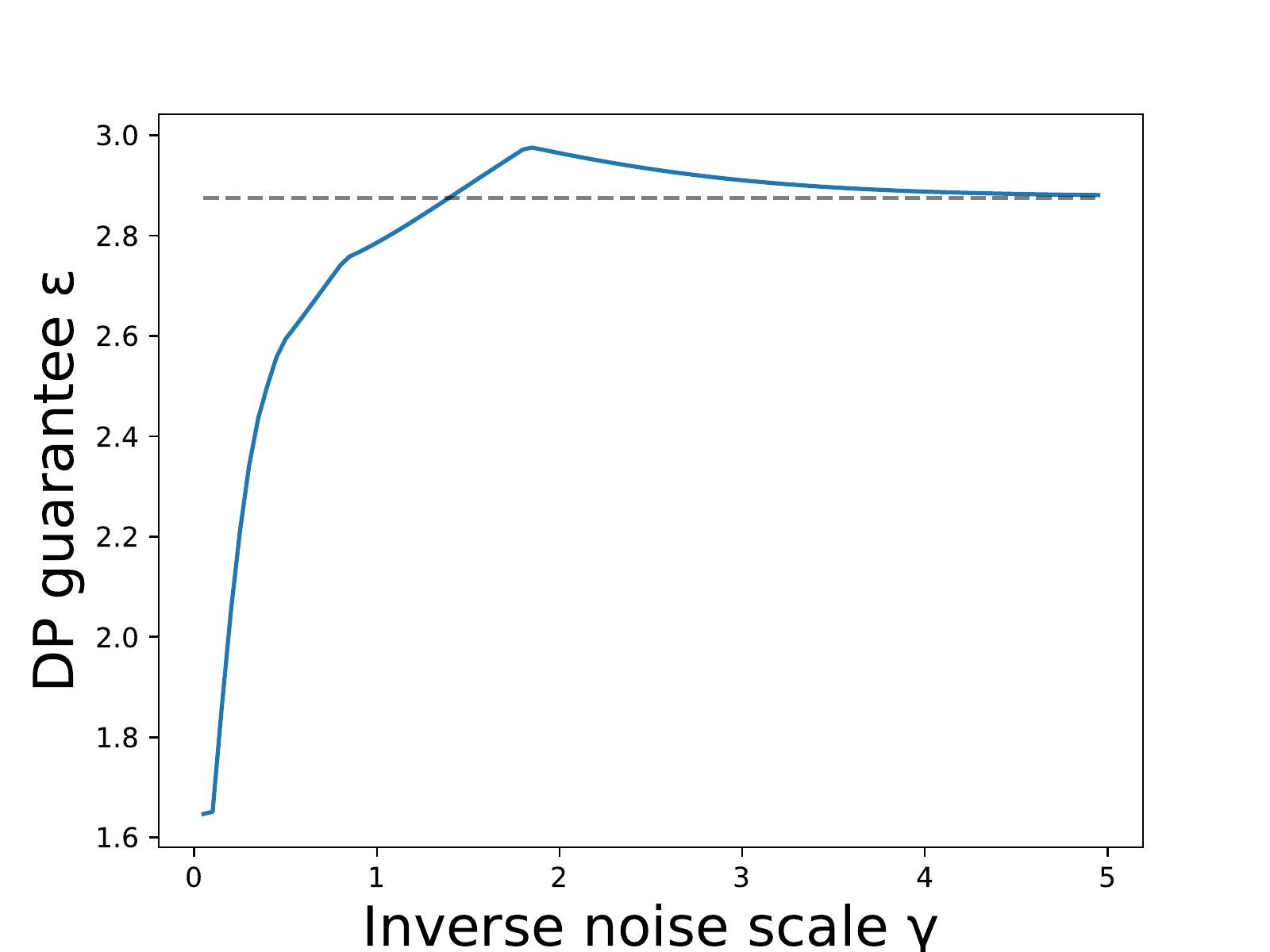}
    \caption{Differential privacy guarantees for MNIST as a function of $\ampNor$, with $\ratioSucTeach=0.9$}
    \label{fig:mnist_dp_gamma}
  \end{minipage}
  \hfill
  \begin{minipage}[b]{0.45\textwidth}
    \includegraphics[width=\textwidth]{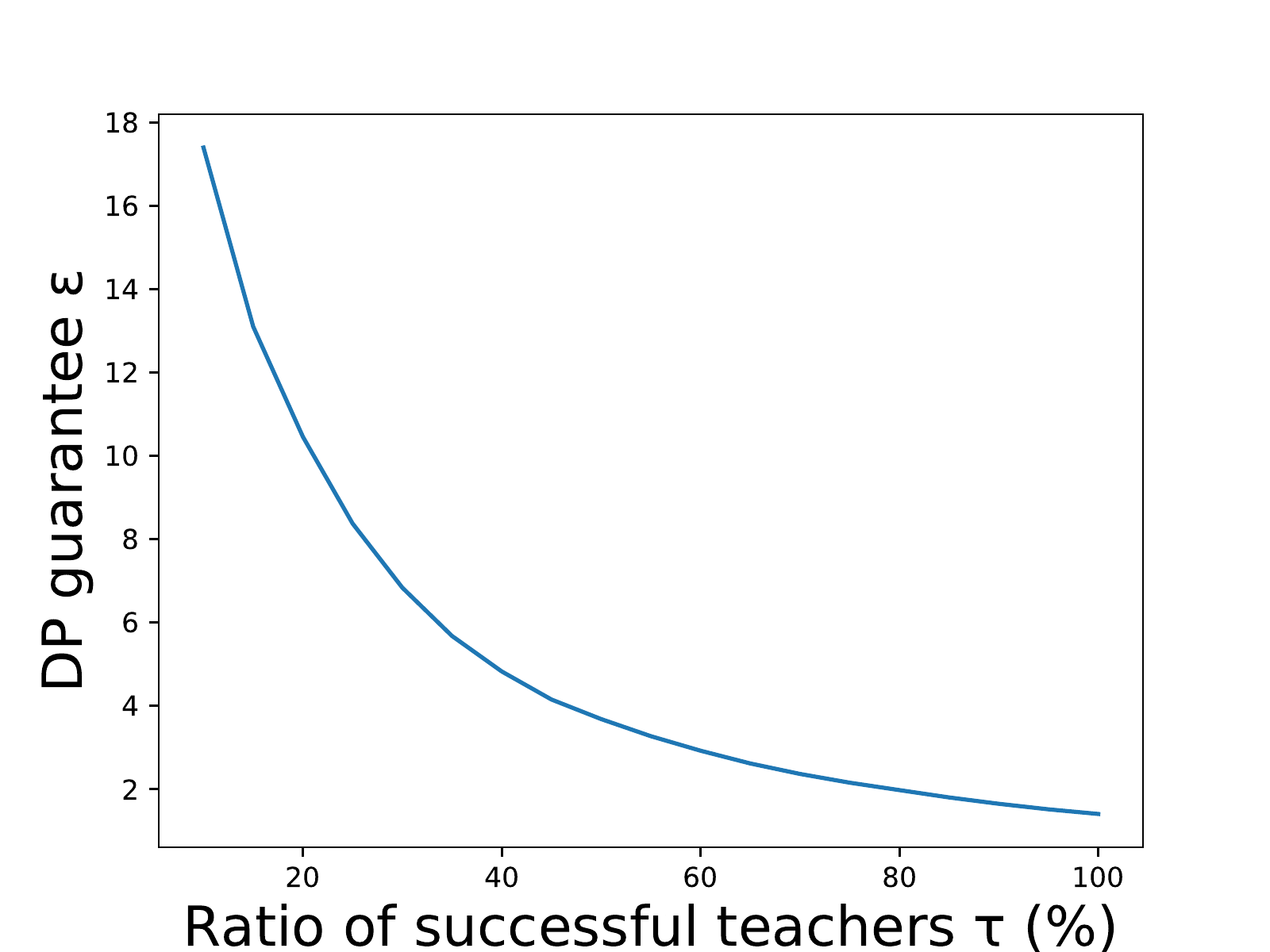}
    \caption{Differential privacy guarantees for MNIST as a function of $\ratioSucTeach$, with $\ampNor=0.1$}
    \label{fig:mnist_dp_tau}
  \end{minipage}
\end{figure}

\paragraph{Performances on SVHN.}
Table~\ref{table:svhn} presents our experimental results on SVHN dataset~\footnote{Note that our DP guarantee $\epsilon$ for Trusted cannot be directly compared with PATE's one since we do not use the same $\delta$.}. The variance on the accuracy is much smaller than for MNIST dataset because the test set is constituted of $16032$ samples. Similarly to the MNIST experiment, the accuracy and the privacy cost increase when less noise is applied because less teachers noised their votes (i.e. when $\ratioSucTeach$ is small).
The DP guarantees are not as good as for MNIST, this is due to the high amount of queries ($500$) necessary to obtain a good accuracy because the learning task is more complex.

\begin{table}[ht]
  \centering
  \caption{SVHN experimental results for $500$ queries, with noise inverse scale $\gamma = 0.1$, $\delta = 10^{-5}$}
  \begin{tabular}{lccc}
        \hline\noalign{\smallskip}
        Framework & $\epsilon$ & Acc. [\%]  & HE overhead \\
        \noalign{\smallskip}\hline\noalign{\smallskip}
        Non-private      & -      & $84.7$ & - \\
        Trusted          & $4.73$ & $83.7$ & - \\
        \noalign{\smallskip}\hline\noalign{\smallskip}
        $\ratioSucTeach=1$   & $4.73$ & $83.5$ & \multirow{4}{*}{$32.5 \; $min} \\
        $\ratioSucTeach=0.9$ & $5.59$ & $83.8$ &  \\
        $\ratioSucTeach=0.7$ & $8.16$ & $84.6$ &  \\
        \noalign{\smallskip}\hline
    \end{tabular}
    \label{table:svhn}
\end{table}

\section{Conclusion and open questions for further works}
\label{sec:conclusion}
Our framework allows a group of agents to collaborate and put together their sensitive knowledge while protecting it via two complementary technologies - differential privacy and homomorphic encryption - against any entity contributing to the learning or having access to the final model. Crucially, our experiments showed that our method is practical for deep learning applications, combining high accuracy, mild computational overhead and privacy guarantees adapting to the number of malicious teachers.

An interesting further work could investigate the fault tolerance of the privacy guarantees with other noises (e.g. Gaussian noise) or other infinite divisions (Laplace distribution can also be infinitely divided using individual Gaussian noises or individual Laplace noises~\cite{goryczka2013secure}).
A more ambitious direction towards collaborative deep learning with privacy would be to design new aggregation operators, more suitable to FHE performances yet still providing good DP bounds. In particular, a linear or quadratic aggregation operator would be amenable to almost negligible homomorphic computations overhead. This lighter homomorphic layer would enable to extend the applicability of our framework to more complex datasets. Such aggregation operators would also allow to associate homomorphic calculations with verifiable computing techniques (e.g.~\cite{fiore2014efficiently}) whereby the server would provide an encrypted aggregation result along with a formal proof that aggregation was indeed done correctly. These perspectives would then allow to address threats beyond the honest-but-curious model.

\section{Declarations}
\subsection{Funding}
The experiments were performed using HPC resources of FactoryIA partially funded by Ile-de-France French region – project SESAME 2017.

\subsection{Conflicts of interest/Competing interests}
Not applicable.

\subsection{Availability of data and material}
The MNIST~\cite{lecun1998mnist} and SVHN~\cite{netzer2011reading} datasets can be found respectively at \url{http://yann.lecun.com/exdb/mnist/} and \url{http://ufldl.stanford.edu/housenumbers/}.

\subsection{Code availability}
The source codes used to run the experiments and compute the DP guarantees can be accessed on \url{https://github.com/Arnaud-GS/SPEED}.

\bibliographystyle{unsrt}

\bibliography{biblio.bib}

\newpage

\appendix
\appendixpage

\setcounter{theorem}{3}
\setcounter{proposition}{4}
\setcounter{definition}{5}

\section{DP analysis of the learning procedure}
In this section, we describe the procedure that computes the overall DP guarantees of the student model learning stage. We summarise this procedure in Section~\ref{sec:algo_dp}, and demonstrate the theorems we use in Sections~\ref{sec:dp_per_query} and~\ref{sec:upperbound_proba_q}.

We call $\RandMech$ the aggregation mechanism that outputs the argmax of the noisy counts. $\RandMech(d, Q)$ is the output of $\RandMech$ for the database $d$ and the query $Q$.

Let $\ampNor\in \mathbb{R}^*_+$ be the inverse scale parameter of the distributed noise. Considering the DP guarantees from the point of view of an entity $\mathcal{E}$, let $\ratioSucTeach\in (0, 1)$ be the ratio of the teachers whose noise is ignored by $\mathcal{E}$. Typically, from the point of view of a colluding teacher, $\ratioSucTeach$ is the ratio of the teachers who do not collude.

\subsection{Analysis algorithm}
\label{sec:algo_dp}
Let us suppose that for every query $Q$ from the student model, we have a privacy guarantee using Theorem~\ref{th:dp_per_query} and that we can upperbound the probability $\mathbb{P}[\RandMech(d;Q)\neq k^*]$ that $\RandMech$ outputs some specific output $k^*$ (in practice we choose $k^*$ to be the unnoisy argmax).
Then, Theorem~\ref{th:AlgoMomentAccountance} gives us an upperbound on the moments accountant per query\footnote{Note that only the third value over which the minimum is taken in Theorem~\ref{th:AlgoMomentAccountance} is data-dependent and, as such, requires this upperbound of $\mathbb{P}[\RandMech(d;Q)\neq k^*]$.}. The computation of these building blocks is detailed in Sections~\ref{sec:dp_per_query} and \ref{sec:upperbound_proba_q}, and the procedure is summarised in Algorithm~\ref{algo:dp_guarantee}.

Let us recall the definition of the moments accountant.
\begin{customdef}{5}
The \emph{moments accountant} of a mechanism $\mathcal{M}$ is defined for any $l\in \mathbb{R_+^*}$ as
\begin{equation*}
    \alpha_{\mathcal{M}}(l):=\max_{\aux,d,d'}\alpha_{\mathcal{M}}(l;\aux,d,d') 
\end{equation*}
where the maximum is taken over any auxiliary input $\aux$ and any pair of adjacent databases $d$, $d'$ and $\alpha_{\mathcal{M}}(l;\aux,d,d'):=\log \left(\mathbb{E}\left[\exp(l C(\mathcal{M},\aux,d,d'))\right]\right)$ is the moment generating function of the privacy loss random variable.
\end{customdef}

\begin{customthm}{3}[\cite{papernot2016semisupervised}]
    Let $\epsilon, l\in \mathbb{R}_+^*$. Let $\mathcal{M}$ be a $(\epsilon,0)$-differentially private mechanism and $q \geq \mathbb{P}[\mathcal{M}(d)\neq k^*]$ for some outcome $k^*$. If $q < \frac{e^{\epsilon} -1 }{e^{2\epsilon} -1}$, then for any $\aux$ and any pair $d$, $d'$ of adjacent databases, $\mathcal{M}$ satisfies
    \begin{align*}
        &\alpha_{\mathcal{M}}(l;\aux,d,d') \leq \min\left(\epsilon l, \frac{\epsilon^2 l (l+1)}{2}, \log\left((1 -q)\left(\frac{1 -q}{1 -e^{\epsilon}q}\right)^l + qe^{\epsilon l}\right)\right).
    \end{align*}
\end{customthm}

Using the moments accountant per query, we evaluate the overall moments accountant by composability, applying the following theorem from~\cite{abadi2016deep}.

\begin{theorem}[\cite{abadi2016deep}]
     \label{theorem:abadi1}
      Let $p\in \mathbb{N}^*$. Let us consider a mechanism $\mathcal{M}$ defined on a set $\mathcal{D}$  that consists of a sequence of adaptive mechanisms $\mathcal{M}_1, \dots, \mathcal{M}_p$ where, for any $i\in [p]$, $\mathcal{M}_i\colon \prod_{j=1}^{i-1}\mathcal{R}_j\times \mathcal{D} \mapsto \mathcal{R}_i$. Then, for any $l\in \mathbb{R_+^*}$,
    \begin{align*}
        \alpha_{\mathcal{M}}(l) \leq \sum_{i=1}^p \alpha_{\mathcal{M}_i}(l).
    \end{align*}
\end{theorem}

Finally, parameter $\delta$ being chosen, the privacy guarantee is derived from the overall moments accountant applying the tail bound property, stated in Theorem~\ref{th:tail_bound} from~\cite{abadi2016deep}.

\begin{theorem}[\cite{abadi2016deep}]
    \label{th:tail_bound}
    For any $\epsilon \in \mathbb{R}_+^*$, the mechanism $\mathcal{M}$ is $(\epsilon, \delta)$-differentially private for
    \begin{align*}
        \delta = \min_{l\in \mathbb{N}^*} \exp(\alpha_{\mathcal{M}}(l) - l\epsilon).
    \end{align*}
\end{theorem}

\begin{algorithm}[H]
\label{algo:dp_guarantee}
\SetAlgoLined
\SetKwInOut{KwIn}{Input}
\SetKwInOut{KwOut}{Output}
\KwIn{number of teachers $n$, number of classes $K$, ratio $\tau$ of teachers with secret noise, set of queries $\mathcal{Q}$, unnoisy teachers' counts $n_k$, inverse noise scale $\gamma$, $l_{max}$~\footnote{To determine the DP guarantees presented in the paper, we took $l_{max}=25$ because it seems empirically that it captures the best moments accountant in every case.}, $\delta$}
\KwOut{$\epsilon$}
\For{$l$ in $[l_{max}]$}{
 $\alpha(l) \leftarrow 0$ \\
 \For{query $Q$ in $\mathcal{Q}$}{
  Compute the privacy cost of $Q$ and an upperbound of $\mathbb{P}[\RandMech(d;Q)\neq k^*]$\;
  Derive the moments accountant $\alpha_Q(l)$ with Theorem~\ref{th:AlgoMomentAccountance}\;
  $\alpha(l) \leftarrow \alpha(l) + \alpha_Q(l)$\;
  }
  $\epsilon(l) \leftarrow \frac{\alpha(l) - \delta}{l}$\;
 }
 $\epsilon \leftarrow \min_{l\in [l_{max}]} \epsilon(l)$\;
 \caption{Algorithm to determine the overall privacy guarantee of the learning procedure}
\end{algorithm}

\subsection{DP guarantee per query in the BHBC framework}
\label{sec:dp_per_query}
\paragraph{Preliminaries on the generalised Laplace distribution.}
For every teacher $j$ who did send noise and whose noise is secret, the noise sent by $j$ is distributed as $G_1^{(j)} - G_2^{(j)}$ where $G_1^{(j)}$ and $G_2^{(j)}$ are two i.i.d. random variables with gamma density $u \mapsto \frac{1}{\left(\frac{1}{\ampNor}\right)^{\frac{1}{n}}\Gamma\left(\frac{1}{n}\right)}u^{\frac{1}{n}-1}e^{-\ampNor u}$ and characteristic function $t\mapsto \left(\frac{1}{1-i\frac{t}{\ampNor}}\right)^{\frac{1}{n}}$ (see~\cite{kotz2001laplace}). Hence, the characteristic function of $G_1^{(j)} - G_2^{(j)}$ is $\psi\colon t\mapsto \left(\frac{1}{1+\left(\frac{t}{\ampNor}\right)^2}\right)^{\frac{1}{n}}$.
By summing over all the teachers who did send a secret noise, we get a total noise whose characteristic function is $\psi^{\ratioSucTeach n}\colon t\mapsto \left(\frac{1}{1+\left(\frac{t}{\ampNor}\right)^2}\right)^{\ratioSucTeach}$. The corresponding moment generating function is $t\mapsto \left(\frac{1}{1-\left(\frac{t}{\ampNor}\right)^2}\right)^{\ratioSucTeach}$. According to~\cite{mathai1993noncentral}, this is the moment generating function of a generalised Laplace distribution whose density is
\begin{align*}
    \density_{\ampNor, \ratioSucTeach} \colon u\in \mathbb{R}^* \mapsto 
         \begin{cases}
           \frac{1}{\left(\frac{1}{\ampNor}\right)^{2\ratioSucTeach}\Gamma\left(\ratioSucTeach\right)^2}e^{\ampNor u}\int_u^{+\infty}t^{\ratioSucTeach-1}\left(t-u\right)^{\ratioSucTeach-1}e^{-2\ampNor t}dt &\quad \text{if} \quad u > 0\\
           \frac{1}{\left(\frac{1}{\ampNor}\right)^{2\ratioSucTeach}\Gamma\left(\ratioSucTeach\right)^2}e^{\ampNor u}\int_0^{+\infty}t^{\ratioSucTeach-1}\left(t-u\right)^{\ratioSucTeach-1}e^{-2\ampNor t}dt &\quad \text{if} \quad u < 0\\
         \end{cases}
\end{align*}
which is actually
\begin{align*}
    u\in \mathbb{R}^* \mapsto &\frac{1}{\left(\frac{1}{\ampNor}\right)^{2\ratioSucTeach}\Gamma\left(\ratioSucTeach\right)^2}e^{\ampNor \abs{u}}\int_{\abs{u}}^{+\infty}t^{\ratioSucTeach-1}\left(t-\abs{u}\right)^{\ratioSucTeach-1}e^{-2\ampNor t}dt\\
    &= \frac{\ampNor^{2\ratioSucTeach-1}}{\Gamma\left(\ratioSucTeach\right)^2}e^{\ampNor \abs{u}}\int_0^{+\infty}\left(\frac{v}{\ampNor}+\abs{u}\right)^{\ratioSucTeach-1}\left(\frac{v}{\ampNor}\right)^{\ratioSucTeach-1}e^{-2(v+\ampNor\abs{u})}dv\\
    &\qquad \qquad \text{(by the substitution $v=\ampNor(t-\abs{u})$)}\\
    &= \normalInt_{\ampNor, \ratioSucTeach} e^{-\ampNor \abs{u}}I_{\ratioSucTeach}(\ampNor \abs{u})
\end{align*}
where $I_{\ratioSucTeach} \colon v \in \mathbb{R^*_+} \mapsto \int_0^{+\infty}\left(x+v\right)^{\ratioSucTeach-1}x^{\ratioSucTeach-1}e^{-2x}dx$ and $\normalInt_{\ampNor, \ratioSucTeach} = \frac{\ampNor}{\Gamma\left(\ratioSucTeach\right)^2}$.

Let us remark that, since $\ratioSucTeach-1 \leq 0$, $I_{\ratioSucTeach}$ is decreasing on $\mathbb{R^*_+}$.

As a density function, $\density_{\ampNor, \ratioSucTeach}$ is integrable on $\mathbb{R}$ (it can also be proved using Lemma~\ref{lemma:maj_I}). We call $\cumulDensity_{\ampNor, \ratioSucTeach}$ the associated cumulative distribution function:
\begin{align*}
    \cumulDensity_{\ampNor, \ratioSucTeach}\colon t\in \mathbb{R}\mapsto \int_{-\infty}^t \density_{\ampNor, \ratioSucTeach}(u)du
\end{align*}

Note that, $\underset{+\infty}{\lim}\cumulDensity_{\ampNor, \ratioSucTeach} = 1$ and, since $\density_{\ampNor, \ratioSucTeach}$ is pair, $\cumulDensity_{\ampNor, \ratioSucTeach}(0) = \frac{1}{2}$ and
\begin{align}
    \label{eq:imparity_F_-_1}
    \forall t\in \mathbb{R}, \cumulDensity_{\ampNor, \ratioSucTeach}(t) + \cumulDensity_{\ampNor, \ratioSucTeach}(-t) = 1.
\end{align}

If there is no ambiguity on the parameters $\ampNor$ and $\ratioSucTeach$, we will only write $\density$, $\cumulDensity$, $I$ and $\normalInt$.

\begin{lemma}
    \label{lemma:mimic_dwork}
    Let $r$ be a random variable following the generalised Laplace distribution as defined above.
    Suppose that we can find a function $M$ of $\ampNor$ and $\ratioSucTeach$ such that, for any $t\in \mathbb{R}$, $\frac{\mathbb{P}[r \geq t]}{\mathbb{P}[r \geq t+2]} \leq M$.

    Then $\RandMech$ is $(\log(M), 0)$-differentially private.
\end{lemma}

\begin{proof}
    We will mimic the proof of the privacy guarantee of the report noisy max from~\cite{dwork2014algorithmic} (Claim 3.9), but with two key adaptations.
    
    First of all, let us warn that our definition of the adjacence of two databases is different from the one of~\cite{dwork2014algorithmic}. Changing one teacher is analogous to changing one individual in the \emph{counting queries} context. This is why the hypotheses must be adapted.
    Indeed, $d$ and $d'$ being two adjacent databases (in our sense), since at most one teacher will change its vote between $d$ and $d'$, we have the property $\abs{n_k - n'_k} \leq 1$ for any $k\in [K]$ but we do not have the property of monotonicity of the counts used in~\cite{dwork2014algorithmic}~\footnote{We could have consider a database $\tilde{d}$ such that $d$ is adjacent to $\tilde{d}$ and $d'$ is adjacent to $\tilde{d}$ with Dwork's definition. Then we could have applied twice the result of~\cite{dwork2014algorithmic} (using $M$ instead of $e^{\ampNor}$ as upper bound of $\frac{\mathbb{P}[r \geq t]}{\mathbb{P}[r \geq t+2]}$ for $(d, \tilde{d})$ and $(\tilde{d}, d')$). Nevertheless, we performed numerical experimentations that make us believe that it would have given worse privacy guarantees than the present result.}.
    
    The second difference is that, $r$ being a random variable following the generalised Laplace distribution, we have to substitute the classical upperbound $e^{2\ampNor}$ (valid for the Laplace distribution) of $\frac{\mathbb{P}[r \geq t]}{\mathbb{P}[r \geq t+2]}$ by $M$.

    We consider a query $Q$.
    Let $k_0\in [K]$.
    
    For any event $E$, we write $\mathbb{P}[E | r_{-k_0}]$ the probability of $E$ under the condition that the draw from the $(K-1)$-dimensional generalised Laplace distribution, used for all the noisy counts except the $k_0$-th count, is equal to $r_{-k_0}$.
    We now suppose this draw $r_{-k_0}$ fixed.
    
    We define $r^* = \min \{r_{k_0} | \forall k\in [K]\setminus \{k_0\}, n_{k_0} + r_{k_0} \geq n_k + r_k\}$.
    Note that, whatever is the tie-breaking policy, $r_{-k_0}$ being fixed, $k_0$ is the output of $\RandMech$ for database $d$ if $r_{k_0} > r^*$ and $k_0$ is not the output of $\RandMech$ if $r_{k_0} < r^*$. Since $\mathbb{P}[r_{k_0}=r^*]=0$, we have $\mathbb{P}[\RandMech(d, Q)=k_0 | r_{-k_0}] = \mathbb{P}[r_{k_0}>r^*] = \mathbb{P}[r_{k_0}\geq r^*]$.
    Moreover, for all $k\in [K]\setminus \{k_0\}$,
    \begin{align*}
        n'_{k_0} + r^* + 2 &\geq n_{k_0} + r^* + 1 &\text{(because $\abs{n_{k_0} - n'_{k_0}} \leq 1$)}\\
        &\geq n_k + r_k + 1 &\text{(by definition of $r^*$)}\\
        &\geq n'_k + r_k &\text{(because $\abs{n_{k_0} - n'_{k_0}} \leq 1$)}
    \end{align*}
    We deduce that, if $r_{k_0} > r^* + 2$, then $k_0$ is the output of $\RandMech$ for database $d'$. Therefore, $\mathbb{P}[\RandMech(d', Q)=k_0 | r_{-k_0}] \geq \mathbb{P}[r_{k_0}>r^*+2] = \mathbb{P}[r_{k_0}\geq r^*+2]$.
    
    Since $\mathbb{P}[r_{k_0}\geq r^*] \leq M \mathbb{P}[r_{k_0}\geq r^*+2]$ by assumption, we can deduce that $\mathbb{P}[\RandMech(d, Q)=k_0 | r_{-k_0}] \leq M \mathbb{P}[\RandMech(d', Q)=k_0 | r_{-k_0}]$.
    This being true for any draw $r_{-k_0}$, the law of total probability gives us $\mathbb{P}[\RandMech(d, Q)=k_0] \leq M \mathbb{P}[\RandMech(d', Q)=k_0]$.
    
    As $d$ and $d'$ play perfectly symmetric roles (unlike in the proof of the report noisy max guarantee from~\cite{dwork2014algorithmic}), we also have $\mathbb{P}[\RandMech(d', Q)=k_0] \leq M \mathbb{P}[\RandMech(d, Q)=k_0]$.
    Since this is true for any query $Q$, we can conclude that $\RandMech$ is $(\log(M), 0)$-differentially private.
\end{proof}

By definition of $\cumulDensity$, $r$ being a random variable following the generalised Laplace distribution, for all $t\in \mathbb{R}$,
\begin{align*}
    \mathbb{P}[r \geq t] = 1 - \cumulDensity(t).
\end{align*}

Let $\dbdist\in \mathbb{R}^*_+$.

In the following, we exhibit upper bounds of $g\colon t\in \mathbb{R} \mapsto \frac{1-\cumulDensity(t)}{1-\cumulDensity(t+\dbdist)}$ (Propositions~\ref{prop:maj_g_gamma} and \ref{prop:maj_g_tau}) to derive privacy guarantees for $\RandMech$ (Theorem~\ref{th:dp_per_query}) taking $\dbdist = 2$.
Let us first state some useful lemmas.

\begin{lemma}
    \label{lemma:log_concavity_of_I}
    Let $\Idist\in \mathbb{R}_+$.
    The application $h\colon \varinI\in \mathbb{R_+^*} \mapsto \frac{I(\varinI)}{I(\varinI+\Idist)}$ is decreasing.
\end{lemma}

\begin{proof}
    We will prove that $h$ is differentiable and that its derivative is non-positive.
    
    Let $\phi \colon (\varinI, t)\in (\mathbb{R}_+^*)^2 \mapsto \left(t+\varinI\right)^{\ratioSucTeach-1}t^{\ratioSucTeach-1}e^{-2\ampNor t}$.
    $\phi$ has a partial derivative in the first variable and, for all $(\varinI, t)\in (\mathbb{R}_+^*)^2$, $\frac{\partial \phi}{\partial \varinI}(\varinI, t) = (\ratioSucTeach-1)\left(t+\varinI\right)^{\ratioSucTeach-2}t^{\ratioSucTeach-1}e^{-2\ampNor t}$.
    $\phi$ and $\frac{\partial \phi}{\partial \varinI}$ are continuous in both variables.
    
    Let $\nonNullThres\in \mathbb{R}_+^*$.
    For all $(\varinI, t)\in [\nonNullThres, +\infty) \times \mathbb{R}_+^*$, $\abs{\frac{\partial \phi}{\partial \varinI}(\varinI, t)} \leq \psi(t)$ where $\psi \colon t\in \mathbb{R}_+^* \mapsto (1-\ratioSucTeach)\left(t+\nonNullThres\right)^{\ratioSucTeach-2}t^{\ratioSucTeach-1}e^{-2\ampNor t}$.
    $\psi$ is continuous and integrable on $[\nonNullThres, +\infty)$.
    Applying Leibniz's theorem, we deduce that $I$ is differentiable on $[\nonNullThres, +\infty)$ and that, for all $\varinI \in [\nonNullThres, +\infty)$,
    $I'(\varinI) = \int_{0}^{+\infty}(\ratioSucTeach-1)\left(t+\varinI\right)^{\ratioSucTeach-2}t^{\ratioSucTeach-1}e^{-2\ampNor t}dt$.
    Since this is true for all $\nonNullThres\in \mathbb{R}_+^*$, we know that $I$ is differentiable on $\mathbb{R}_+^*$ and that, for all $\varinI \in \mathbb{R}_+^*$,
    $I'(\varinI) = \int_{0}^{+\infty}(\ratioSucTeach-1)\left(t+\varinI\right)^{\ratioSucTeach-2}t^{\ratioSucTeach-1}e^{-2\ampNor t}dt$.
    As a consequence, $h$ is differentiable on $\mathbb{R}_+^*$ and, for all $\varinI\in \mathbb{R}_+^*$, $h'(\varinI) = \frac{I(\varinI+\Idist)I'(\varinI) - I(\varinI)I'(\varinI+\Idist)}{I(\varinI+\Idist)^2}$.
    
    Let $\varinI\in \mathbb{R}_+^*$.
    \begin{align*}
        &I(\varinI+\Idist)I'(\varinI) - I(\varinI)I'(\varinI+\Idist)\\
        &\,\,\, = \int_{0}^{+\infty}\left(x+\varinI+\Idist\right)^{\ratioSucTeach-1}x^{\ratioSucTeach-1}e^{-2x}dx \times \int_{0}^{+\infty}(\ratioSucTeach-1)\left(y+\varinI\right)^{\ratioSucTeach-2}y^{\ratioSucTeach-1}e^{-2y}dy\\
        &\,\,\,\quad - \int_{0}^{+\infty}\left(y+\varinI\right)^{\ratioSucTeach-1}y^{\ratioSucTeach-1}e^{-2y}dy \times \int_{0}^{+\infty}(\ratioSucTeach-1)\left(x+\varinI+\Idist\right)^{\ratioSucTeach-2}x^{\ratioSucTeach-1}e^{-2x}dx\\
        &\,\,\, = (\ratioSucTeach-1)\left[\int_{0}^{+\infty}\left(x+\varinI+\Idist\right)^{\ratioSucTeach-1}x^{\ratioSucTeach-1}e^{-2x}\int_{0}^{+\infty}\left(y+\varinI\right)^{\ratioSucTeach-2}y^{\ratioSucTeach-1}e^{-2y}dydx\right.\\
        &\,\,\,\quad \left.- \int_{0}^{+\infty}\left(x+\varinI+\Idist\right)^{\ratioSucTeach-2}x^{\ratioSucTeach-1}e^{-2x}\int_{0}^{+\infty}\left(y+\varinI\right)^{\ratioSucTeach-1}y^{\ratioSucTeach-1}e^{-2y}dydx\right]\\
        &\,\,\, = (\ratioSucTeach-1) \left[\int_{0}^{+\infty}\int_{0}^{+\infty}\left(x+\varinI+\Idist\right)^{\ratioSucTeach-1}\left(y+\varinI\right)^{\ratioSucTeach-2}(xy)^{\ratioSucTeach-1}e^{-2(x+y)}dydx\right.\\
        &\,\,\,\quad \left. - \int_{0}^{+\infty}\int_{0}^{+\infty}\left(x+\varinI+\Idist\right)^{\ratioSucTeach-2}\left(y+\varinI\right)^{\ratioSucTeach-1}(xy)^{\ratioSucTeach-1}e^{-2(x+y)}dydx\right]\\
        &\,\,\, = (\ratioSucTeach-1) \int_{0}^{+\infty}\int_{0}^{+\infty}(xy)^{\ratioSucTeach-1}e^{-2(x+y)}\\
        &\,\,\,\qquad\qquad\qquad \times \left[\left(x+\varinI+\Idist\right)^{\ratioSucTeach-1}\left(y+\varinI\right)^{\ratioSucTeach-2} - \left(x+\varinI+\Idist\right)^{\ratioSucTeach-2}\left(y+\varinI\right)^{\ratioSucTeach-1}\right]dydx\\
        &\,\,\, = (\ratioSucTeach-1) \int_{0}^{+\infty}\int_{0}^{+\infty}\left(x+\varinI+\Idist\right)^{\ratioSucTeach-2}\left(y+\varinI\right)^{\ratioSucTeach-2}(xy)^{\ratioSucTeach-1}e^{-2(x+y)}\\
        &\,\,\,\qquad\qquad\qquad \times \left[\left(x+\varinI+\Idist\right) - \left(y+\varinI\right)\right]dydx\\
        &\,\,\, = (\ratioSucTeach-1) \int_{0}^{+\infty}\int_{0}^{+\infty}\left(x+\varinI+\Idist\right)^{\ratioSucTeach-2}\left(y+\varinI\right)^{\ratioSucTeach-2}(xy)^{\ratioSucTeach-1}e^{-2(x+y)}\\
        &\,\,\,\qquad\qquad\qquad \times \left(x + \Idist - y\right)dydx\\
        &\,\,\, \leq (\ratioSucTeach-1) \int_{0}^{+\infty}\int_{0}^{+\infty}\left(x+\varinI+\Idist\right)^{\ratioSucTeach-2}\left(y+\varinI\right)^{\ratioSucTeach-2}(xy)^{\ratioSucTeach-1}e^{-2(x+y)}(x - y)dydx \numberthis \label{maj_f_deriv_1}
        \intertext{(because $\ratioSucTeach-1 \leq 0$ and $\Idist \geq 0$)}
    \end{align*}
    
    Similarly, we show that
    \begin{align*}
        &I(\varinI+\Idist)I'(\varinI) - I(\varinI)I'(\varinI+\Idist)\\
        &\,\,\, = \int_{0}^{+\infty}\left(y+\varinI+\Idist\right)^{\ratioSucTeach-1}y^{\ratioSucTeach-1}e^{-2y}dy \times \int_{0}^{+\infty}(\ratioSucTeach-1)\left(x+\varinI\right)^{\ratioSucTeach-2}x^{\ratioSucTeach-1}e^{-2x}dx\\
        &\,\,\,\quad - \int_{0}^{+\infty}\left(x+\varinI\right)^{\ratioSucTeach-1}x^{\ratioSucTeach-1}e^{-2x}dx \times \int_{0}^{+\infty}(\ratioSucTeach-1)\left(y+\varinI+\Idist\right)^{\ratioSucTeach-2}y^{\ratioSucTeach-1}e^{-2y}dy\\
        &\,\,\, = (\ratioSucTeach-1)\left[\int_{0}^{+\infty}\left(x+\varinI\right)^{\ratioSucTeach-2}x^{\ratioSucTeach-1}e^{-2x}\int_{0}^{+\infty}\left(y+\varinI+\Idist\right)^{\ratioSucTeach-1}y^{\ratioSucTeach-1}e^{-2y}dydx\right.\\
        &\,\,\,\quad \left.- \int_{0}^{+\infty}\left(x+\varinI\right)^{\ratioSucTeach-1}x^{\ratioSucTeach-1}e^{-2x}\int_{0}^{+\infty}\left(y+\varinI+\Idist\right)^{\ratioSucTeach-2}y^{\ratioSucTeach-1}e^{-2y}dydx\right]\\
        &\,\,\, = (\ratioSucTeach-1) \int_{0}^{+\infty}\int_{0}^{+\infty}(xy)^{\ratioSucTeach-1}e^{-2(x+y)}\\
        &\,\,\,\qquad\qquad\qquad \times \left[\left(x+\varinI\right)^{\ratioSucTeach-2}\left(y+\varinI+\Idist\right)^{\ratioSucTeach-1} - \left(x+\varinI\right)^{\ratioSucTeach-1}\left(y+\varinI+\Idist\right)^{\ratioSucTeach-2}\right]dydx\\
        &\,\,\, = (\ratioSucTeach-1) \int_{0}^{+\infty}\int_{0}^{+\infty}\left(x+\varinI\right)^{\ratioSucTeach-2}\left(y+\varinI+\Idist\right)^{\ratioSucTeach-2}(xy)^{\ratioSucTeach-1}e^{-2(x+y)}\\
        &\,\,\,\qquad\qquad\qquad \times (y + \Idist - x)dydx\\
        &\,\,\, \leq (\ratioSucTeach-1) \int_{0}^{+\infty}\int_{0}^{+\infty}\left(x+\varinI\right)^{\ratioSucTeach-2}\left(y+\varinI+\Idist\right)^{\ratioSucTeach-2}(xy)^{\ratioSucTeach-1}e^{-2(x+y)}(y - x)dydx \numberthis \label{maj_f_deriv_2}.
    \end{align*}
    
    Alternatively, we can use~\ref{maj_f_deriv_1} to deduce~\ref{maj_f_deriv_2} directly using Fubini's theorem and exchanging the roles of $x$ and $y$.
    
    From~\ref{maj_f_deriv_1} and~\ref{maj_f_deriv_2}, we get:
    \begin{align*}
        &2\times \left[ I(\varinI+\Idist)I'(\varinI) - I(\varinI)I'(\varinI+\Idist) \right]\\
        &\,\,\, \leq (\ratioSucTeach-1) \int_{0}^{+\infty}\int_{0}^{+\infty}\left(x+\varinI+\Idist\right)^{\ratioSucTeach-2}\left(y+\varinI\right)^{\ratioSucTeach-2}\left(x - y\right)(xy)^{\ratioSucTeach-1}e^{-2(x+y)}dydx\\
        &\,\,\, \,\,\, + (\ratioSucTeach-1) \int_{0}^{+\infty}\int_{0}^{+\infty}\left(x+\varinI\right)^{\ratioSucTeach-2}\left(y+\varinI+\Idist\right)^{\ratioSucTeach-2}\left(y - x\right)(xy)^{\ratioSucTeach-1}e^{-2(x+y)}dydx\\
        &\,\,\, = (\ratioSucTeach-1) \int_{0}^{+\infty}\int_{0}^{+\infty}\left(x - y\right)(xy)^{\ratioSucTeach-1}e^{-2(x+y)}\\
        &\,\,\,\qquad\qquad\qquad \times \left[\left(x+\varinI+\Idist\right)^{\ratioSucTeach-2}\left(y+\varinI\right)^{\ratioSucTeach-2} - \left(x+\varinI\right)^{\ratioSucTeach-2}\left(y+\varinI+\Idist\right)^{\ratioSucTeach-2}\right]dydx.
    \end{align*}
    
    Let $(x, y)\in (\mathbb{R}_+^*)^2$.
    
    Note that $\left(x+\varinI+\Idist\right)\left(y+\varinI\right) - \left(x+\varinI\right)\left(y+\varinI+\Idist\right) = \Idist(y-x)$ and then
    \begin{align*}
        &\left(x+\varinI+\Idist\right)^{\ratioSucTeach-2}\left(y+\varinI\right)^{\ratioSucTeach-2} \geq \left(x+\varinI\right)^{\ratioSucTeach-2}\left(y+\varinI+\Idist\right)^{\ratioSucTeach-2} &\\
        &\quad \Leftrightarrow \left(x+\varinI+\Idist\right)\left(y+\varinI\right) \leq \left(x+\varinI\right)\left(y+\varinI+\Idist\right) &\text{(because $\ratioSucTeach-2 < 0$)}\\
        &\quad \Leftrightarrow x \geq y.
    \end{align*}
    We deduce that 
    \begin{align*}
        \left[\left(x+\varinI+\Idist\right)^{\ratioSucTeach-2}\left(y+\varinI\right)^{\ratioSucTeach-2} - \left(x+\varinI\right)^{\ratioSucTeach-2}\left(y+\varinI+\Idist\right)^{\ratioSucTeach-2}\right]\left(x - y\right) \geq 0.
    \end{align*}
    This inequality being true for all $(x, y)\in (\mathbb{R}_+^*)^2$ and, since $\ratioSucTeach-1 \leq 0$, we have:
    \begin{align*}
        &(\ratioSucTeach-1) \int_{0}^{+\infty}\int_{0}^{+\infty}\left[\left(x+\varinI+\Idist\right)^{\ratioSucTeach-2}\left(y+\varinI\right)^{\ratioSucTeach-2} - \left(x+\varinI\right)^{\ratioSucTeach-2}\left(y+\varinI+\Idist\right)^{\ratioSucTeach-2}\right]\\
        &\qquad\qquad\qquad \times \left(x - y\right)(xy)^{\ratioSucTeach-1}e^{-2(x+y)}dydx \leq 0.
    \end{align*}
    Finally, $I(\varinI+\Idist)I'(\varinI) - I(\varinI)I'(\varinI+\Idist) \leq 0$ and $h'(\varinI)\leq 0$.
    
    Since this is true for any $\varinI\in \mathbb{R}_+^*$, we can conclude that $h$ is decreasing on $\mathbb{R}^*_+$.
\end{proof}

\begin{lemma}
    \label{lemma:max_of_g}
    The function $g$ has a maximum on $\mathbb{R}$, and this maximum is reached in the interval $[-\frac{\dbdist}{2}; 0]$.
\end{lemma}

\begin{proof}
    Since $\density$ is defined on $\mathbb{R}^*$, $\cumulDensity$ is differentiable on $\mathbb{R}^*$. Thus $g$ is differentiable on $\mathbb{R}^*\setminus \{-a\}$ and, for all $t\in \mathbb{R}^*\setminus \{-\dbdist\}$,
    \begin{align*}
        g'(t) = \frac{(1-\cumulDensity(t))\density(t+\dbdist) - (1-\cumulDensity(t+\dbdist))\density(t)}{(1-\cumulDensity(t+\dbdist))^2}.
    \end{align*}
    
    First of all, let us prove that $g$ is increasing on $(-\infty; -\frac{a}{2})$.
    For all $t\in (-\infty; -\dbdist)$, $\abs{t} = -t \geq -t-\dbdist  = \abs{t+\dbdist}$ and, for all $t\in (-\dbdist; -\frac{\dbdist}{2})$, $\abs{t} = -t \geq t+\dbdist  = \abs{t+\dbdist}$.
    Let $t\in (-\infty; -a) \cup (-a; -\frac{a}{2})$.
    Then, since $x\mapsto e^{-\ampNor x}I(\ampNor x)$ is decreasing on $\mathbb{R}^*_+$, $e^{-\ampNor \abs{t}}I(\ampNor \abs{t}) \leq e^{-\ampNor \abs{t+\dbdist}}I(\ampNor \abs{t+a})$ which means $\density(t)\leq \density(t+\dbdist)$.
    Besides, $\cumulDensity$ is increasing then, since $a\geq 0$, $1-\cumulDensity(t+\dbdist) \leq 1-\cumulDensity(t)$.
    Since $\density(t)$, $\density(t+\dbdist)$, $1-\cumulDensity(t)$ and $1-\cumulDensity(t)$ are all positive quantities, we deduce that $g'(t) \geq 0$.
    Then, $g$ is increasing on $(-\infty; -\dbdist)$ and on $(-\dbdist; -\frac{\dbdist}{2})$ and since $g$ is defined and continuous in $-a$, $g$ is increasing on $(-\infty; -\frac{\dbdist}{2})$.
    
    Let us now prove that $g$ is decreasing on $\mathbb{R}_+$.
    Let $t\in \mathbb{R}^*_+$.
    \begin{align*}
        &\frac{(1-\cumulDensity(t+\dbdist))^2}{\normalInt^2}g'(t)\\
        &\quad = \frac{1}{\normalInt^2}\left[(1-\cumulDensity(t))\density(t+\dbdist) - (1-\cumulDensity(t+\dbdist))\density(t)\right]\\
        &\quad = e^{-\ampNor \abs{t+\dbdist}}I(\ampNor \abs{t+\dbdist})\int_t^{+\infty}e^{-\ampNor \abs{u}}I(\ampNor \abs{u})du\\
        &\quad \qquad - e^{-\ampNor \abs{t}}I(\ampNor \abs{t})\int_{t+\dbdist}^{+\infty}e^{-\ampNor \abs{u}}I(\ampNor \abs{u})du\\
        &\quad = e^{-\ampNor (t+\dbdist)}I(\ampNor (t+\dbdist))\int_t^{+\infty}e^{-\ampNor u}I(\ampNor u)du - e^{-\ampNor t}I(\ampNor t)\int_{t+\dbdist}^{+\infty}e^{-\ampNor u}I(\ampNor u)du\\
        &\quad = e^{-\ampNor (t+\dbdist)}I(\ampNor (t+\dbdist))\int_t^{+\infty}e^{-\ampNor u}I(\ampNor u)du\\
        &\quad \qquad - e^{-\ampNor t}I(\ampNor t)\int_t^{+\infty}e^{-\ampNor (v+\dbdist)}I(\ampNor (v+\dbdist))dv\\
        &\qquad \qquad\text{(by the substitution $v = u-\dbdist$)}\\
        &\quad = e^{-\ampNor (t+\dbdist)}\left[\int_t^{+\infty}e^{-\ampNor u}\left[I(\ampNor (t+\dbdist))I(\ampNor u) - I(\ampNor t)I(\ampNor (u+\dbdist))\right]du\right]
    \end{align*}
    
    For any $u\in [t; +\infty)$, Lemma~\ref{lemma:log_concavity_of_I} with $\Idist = \ampNor \dbdist$ tells us that $\frac{I(\ampNor u)}{I(\ampNor (u+\dbdist))} \leq \frac{I(\ampNor t)}{I(\ampNor (t+\dbdist))}$ which means $I(\ampNor (t+\dbdist))I(\ampNor u) - I(\ampNor t)I(\ampNor (u+\dbdist)) \leq 0$.
    
    Therefore, $\int_t^{+\infty}e^{-\ampNor u}\left[I(\ampNor (t+\dbdist))I(\ampNor u) - I(\ampNor t)I(\ampNor (u+\dbdist))\right]du \leq 0$ and finally $g'(t)\leq 0$.
    This being valid for all $t\in \mathbb{R}^*_+$ and $g$ being continuous in $0$, we deduce that $g$ is decreasing on $\mathbb{R}_+$.
    
    From the two previous discussions and from the fact that $g$ is continuous on $[-\frac{\dbdist}{2}; 0]$, we conclude that $g$ has a maximum on $\mathbb{R}$ and that this maximum is reached in $[-\frac{\dbdist}{2}; 0]$.
\end{proof}

\begin{proposition}
    \label{prop:maj_g_gamma}
    For all $t\in [-\frac{\dbdist}{2}; 0]$,
    \begin{align*}
            g(t) \leq 1 + 2\frac{\int_0^{\frac{\ampNor \dbdist}{2}}e^{-v}I(v)dv}{\int_{\ampNor \dbdist}^{+\infty}e^{-v}I(v)dv}.
    \end{align*}
\end{proposition}

\begin{proof}
    For all $t\in [-\frac{\dbdist}{2}; 0]$,
    \begin{align*}
            g(t) = 1 + \frac{\cumulDensity(t+\dbdist)-\cumulDensity(t)}{1-\cumulDensity(t+\dbdist)}.
    \end{align*}
    
    Calling $\phi\colon t\in [-\frac{\dbdist}{2}; 0] \mapsto \cumulDensity(t+\dbdist)-\cumulDensity(t)$, we know that $\phi$ is differentiable on $[-\frac{\dbdist}{2}; 0)$ and that $\phi'\colon t\in [-\frac{\dbdist}{2}; 0) \mapsto \density(t+\dbdist)-\density(t)$.
    Since $x\in \mathbb{R}^*_+\mapsto e^{-x}I(x)$ is decreasing, we have, for all $t\in [-\frac{\dbdist}{2}; 0)$,
    \begin{align*}
        \phi'(t) \geq 0 &\Leftrightarrow e^{-\ampNor \abs{t+\dbdist}}I(\ampNor \abs{t+\dbdist}) \geq e^{-\ampNor \abs{t}}I(\ampNor \abs{t}) &\\
        &\Leftrightarrow \abs{t+\dbdist} \leq \abs{t} &\\
        &\Leftrightarrow t+\dbdist \leq -t &\text{(because $t+\dbdist \geq 0$ and $t\leq 0$)}\\
        &\Leftrightarrow t\leq -\frac{\dbdist}{2}.
    \end{align*}
    
    Since $\phi$ is continuous in $0$, we deduce that $\phi$ is decreasing on $[-\frac{\dbdist}{2}; 0]$ and then, for all $t\in [-\frac{\dbdist}{2}; 0]$, $\cumulDensity(t+\dbdist)-\cumulDensity(t) \leq \cumulDensity(\frac{\dbdist}{2})-\cumulDensity(-\frac{\dbdist}{2})$.
    Moreover, since $\cumulDensity$ is increasing, for all $t\in [-\frac{\dbdist}{2}; 0]$, $1-\cumulDensity(t+\dbdist) \geq 1-\cumulDensity(\dbdist)$.
    
    Finally, for all $t\in [-\frac{\dbdist}{2}; 0]$,
    \begin{align*}
            g(t) &\leq 1 + \frac{\cumulDensity(\frac{\dbdist}{2})-\cumulDensity(-\frac{\dbdist}{2})}{1-\cumulDensity(\dbdist)} &\\
            &= 1 + \frac{\normalInt \int_{-\frac{\dbdist}{2}}^{\frac{\dbdist}{2}}e^{-\ampNor \abs{u}}I(\ampNor \abs{u})du}{\normalInt \int_{\dbdist}^{+\infty}e^{-\ampNor \abs{u}}I(\ampNor \abs{u})du} &\\
            &= 1 + \frac{\frac{\normalInt}{\ampNor} \int_{-\frac{\ampNor \dbdist}{2}}^{\frac{\ampNor \dbdist}{2}}e^{-\abs{v}}I(\abs{v})dv}{\frac{\normalInt}{\ampNor} \int_{\ampNor \dbdist}^{+\infty}e^{-\abs{v}}I(\abs{v})dv} &\text{(by the substitutions $v \!=\! \ampNor u$)}\\
            &= 1 + \frac{\int_{-\frac{\ampNor \dbdist}{2}}^0 e^{-\abs{v}}I(\abs{v})dv + \int_0^{\frac{\ampNor \dbdist}{2}}e^{-\abs{v}}I(\abs{v})dv}{\int_{\ampNor \dbdist}^{+\infty}e^{-\abs{v}}I(\abs{v})dv} &\\
            &= 1 + \frac{\int_0^{\frac{\ampNor \dbdist}{2}} e^{-\abs{v'}}I(\abs{v'})dv' + \int_0^{\frac{\ampNor \dbdist}{2}}e^{-\abs{v}}I(\abs{v})dv}{\int_{\ampNor \dbdist}^{+\infty}e^{-\abs{v}}I(\abs{v})dv} &\text{(by the substitution $v' \!=\! -v$)}\\
            &= 1 + \frac{2\int_0^{\frac{\ampNor \dbdist}{2}}e^{-\abs{v}}I(\abs{v})dv}{\int_{\ampNor \dbdist}^{+\infty}e^{-\abs{v}}I(\abs{v})dv} &\\
            &= 1 + 2\frac{\int_0^{\frac{\ampNor \dbdist}{2}}e^{-v}I(v)dv}{\int_{\ampNor \dbdist}^{+\infty}e^{-v}I(v)dv}.
    \end{align*}
\end{proof}

\begin{proposition}
    \label{prop:maj_g_tau}
    Let us suppose that $\ratioSucTeach > \frac{1}{2}$.
    
    For all $t\in [-\frac{\dbdist}{2}; 0]$,
    \begin{align*}
            g(t) \leq g(0) - \frac{\dbdist}{2}g'(0)
    \end{align*}
    with
    \begin{align*}
        g'(0) = \ampNor \frac{\frac{\Gamma(\ratioSucTeach)^2}{2}e^{-\ampNor \dbdist}I(\ampNor \dbdist) - I(0)\int_{\ampNor \dbdist}^{+\infty}e^{-v}I(v)dv}{\left(\int_{\ampNor \dbdist}^{+\infty}e^{-v}I(v)dv\right)^2}.
    \end{align*}
\end{proposition}

\begin{proof}
 The result basically comes from the fact that $g$ is concave on $[\argmax(g); 0]$ which we prove hereafter.
 
 From the proof of Lemma~\ref{lemma:max_of_g} we know that $g$ is differentiable on $[-\frac{\dbdist}{2}; 0)$ and $g'\colon t\mapsto \frac{(1-\cumulDensity(t))\density(t+\dbdist) - (1-\cumulDensity(t+\dbdist))\density(t)}{(1-\cumulDensity(t+\dbdist))^2} = \frac{g(t)\density(t+\dbdist) - \density(t)}{1-\cumulDensity(t+\dbdist)}$. In the proof of Lemma~\ref{lemma:log_concavity_of_I}, we saw that $I$ is differentiable on $\mathbb{R}^*_+$ and thus $\density$ is differentiable on $\mathbb{R}^*_+$. Finally, we get that $g'$ is differentiable on $(-\dbdist; 0)$ and, for all $t\in (-\dbdist; 0)$,
 \begin{align*}
     g''(t) &= \frac{1}{(1-\cumulDensity(t+\dbdist))^2}\left[(1-\cumulDensity(t+\dbdist))[g'(t)\density(t+\dbdist) + g(t)\density'(t+\dbdist) - \density'(t)]\right.\\
     &\qquad\qquad\qquad\qquad\qquad \left. + \density(t+\dbdist)[g(t)\density(t+\dbdist)-\density(t)]\right]\\
     &= \frac{1}{(1-\cumulDensity(t+\dbdist))^2}\left[(1-\cumulDensity(t+\dbdist))[g'(t)\density(t+\dbdist) + g(t)\density'(t+\dbdist) - \density'(t)]\right.\\
     &\qquad\qquad\qquad\qquad\qquad \left. + (1-\cumulDensity(t+\dbdist))\density(t+\dbdist)g'(t)\right]\\
     &= 2g'(t)\frac{\density(t+\dbdist)}{1-\cumulDensity(t+\dbdist)} + \frac{(1-\cumulDensity(t+\dbdist))[g(t)\density'(t+\dbdist) - \density'(t)]}{(1-\cumulDensity(t+\dbdist))^2}\\
     &= 2g'(t)\frac{\density(t+\dbdist)}{1-\cumulDensity(t+\dbdist)} + \frac{(1-\cumulDensity(t))\density'(t+\dbdist) - (1-\cumulDensity(t+\dbdist))\density'(t)}{(1-\cumulDensity(t+\dbdist))^2}.
 \end{align*}
 Since $I'$ is strictly negative on $\mathbb{R}^*_+$, for all $u < 0$, $\density'(u) = \normalInt \ampNor [e^{\ampNor u}I(-\ampNor u) - e^{\ampNor u}I'(-\ampNor u)] > 0$ and, for all $u > 0$, $\density'(u) = \normalInt \ampNor [-e^{-\ampNor u}I(\ampNor u) + e^{-\ampNor u}I'(\ampNor u)] < 0$.
 Then, for all $t\in (-\dbdist; 0)$, $\density'(t) > 0$ and $\density'(t+\dbdist) < 0$ and, since $1-\cumulDensity(t) > 0$ and $1-\cumulDensity(t+\dbdist) > 0$, $(1-\cumulDensity(t))\density'(t+\dbdist) < 0$ and $(1-\cumulDensity(t+\dbdist))\density'(t) > 0$.
 We deduce that, for all $t\in (-\dbdist; 0)$,
 \begin{align}
    \label{ineq:g''_g'}
     g''(t) < 2g'(t)\frac{\density(t+\dbdist)}{1-\cumulDensity(t+\dbdist)} + \frac{(1-\cumulDensity(t))\density'(t+\dbdist)}{(1-\cumulDensity(t+\dbdist))^2}
 \end{align}
 where $2 \frac{\density(t+\dbdist)}{1-\cumulDensity(t+\dbdist)} > 0$  and $\frac{(1-\cumulDensity(t))\density'(t+\dbdist)}{(1-\cumulDensity(t+\dbdist))^2} < 0$.
 
 According to Lemma~\ref{lemma:max_of_g}, $g$ has a maximum, which is reached on $[-\frac{\dbdist}{2}; 0]$. Let $t_{max} = \argmax(g)$.
 If $t_{max} \neq 0$, we can argue that $g'(t_{max}) = 0$ and then, from Inequation~\ref{ineq:g''_g'}, $g''$ is strictly negative on a neighbourhood of $t_{max}$. This implies that $g'$ is decreasing on a neighbourhood of $(t_{max})^+$ and then strictly negative on a neighbourhood of $(t_{max})^+$.
 
 Removing the assumption that $t_{max} \neq 0$, we need to be slightly more subtle since $g'$ is not differentiable in $0$ (because $I$ is not differentiable in $0$).
 
 Since $\ratioSucTeach > \frac{1}{2}$, $v\mapsto v^{2\ratioSucTeach-2}e^{-2v}$ is integrable on $\mathbb{R}^*_+$ and we can extend the definition of $I$ to $\mathbb{R}_+$. This implies in particular that $\cumulDensity$ and then $g$ are differentiable on the whole interval $(-a; +\infty)$ (with $g'(0) = \frac{(1-\cumulDensity(0))\density(\dbdist) - (1-\cumulDensity(\dbdist))\density(0)}{(1-\cumulDensity(\dbdist))^2}$).
 Then $g'(t_{max}) = 0$ and, from Inequation~\ref{ineq:g''_g'}, $\underset{(t_{max})^+}{\lim} g'' < \frac{(1-\cumulDensity(t_{max}))\density'(t_{max}+\dbdist)}{(1-\cumulDensity(t_{max}+\dbdist))^2} < 0$. Thus $g''$ (not defined in $0$) is strictly negative on a neighbourhood of $(t_{max})^+$. Then $g'$ is strictly decreasing on a neighbourhood of $(t_{max})^+$ and, by continuity in $t_{max}$, strictly negative on a neighbourhood of $(t_{max})^+$.
 
 Let us suppose that $g''(t) \geq 0$ for a $t$ in $[t_{max}; 0)$ (trivially false if $t_{max} = 0$ since $[t_{max}; 0)$ is empty in this case). We fix such a $t$ and call it $t_0$. Then, from Inequation~\ref{ineq:g''_g'}, $g'(t_0) > 0$ and we can fix $t_1 = \inf\{t\in [t_{max}; t_0] | g'(t) \geq 0\}$. $g'$ is non-negative on a neighbourhood of $(t_1)^+$ thus $t_1 > t_{max}$. We also know that $g'$ is non-positive on $[t_{max}; t_1)$ by definition of $t_1$. This implies $g'(t_1) = 0$. Since $g'(t_1)=0$, from Inequation~\ref{ineq:g''_g'}, we know that $g''(t_1)<0$ and then $g'$ is strictly negative on a neighbourhood of $(t_1)^+$. We get a contradiction so $g''(t) < 0$ for all $t\in [t_{max}; 0)$. We deduce that $g'$ is decreasing on $[t_{max}; 0)$.
 
 Thus, for all $t\in [t_{max}; 0)$, $g'(t) \geq g'(0)$. As a consequence, since $t_{max}\leq 0$, $g(t_{max}) \leq g(0) + t_{max}g'(0)$. Besides, $t_{max} \geq -\frac{\dbdist}{2}$ and $g'(0) \leq g'(t_{max}) = 0$, thus $g(t_{max}) \leq g(0) - \frac{\dbdist}{2}g'(0)$.
 Finally, by definition of $t_{max}$, for all $t\in \mathbb{R}$,
 \begin{align*}
     g(t) \leq g(0) - \frac{\dbdist}{2}g'(0)
 \end{align*}
 with
 \begin{align*}
     g'(0) &= \frac{(1-\cumulDensity(0))\density(\dbdist) - (1-\cumulDensity(\dbdist))\density(0)}{(1-\cumulDensity(\dbdist))^2} &\\
     &= \frac{\frac{1}{2}\normalInt e^{-\ampNor \dbdist}I(\ampNor \dbdist) - \normalInt^2I(0)\int_{\dbdist}^{+\infty}e^{-\ampNor u}I(\ampNor u)du}{\left(\normalInt\int_{\dbdist}^{+\infty}e^{-\ampNor u}I(\ampNor u)du\right)^2} &\\
     &= \frac{\frac{1}{2\normalInt}e^{-\ampNor \dbdist}I(\ampNor \dbdist) - I(0)\int_{\dbdist}^{+\infty}e^{-\ampNor u}I(\ampNor u)du}{\left(\int_{\dbdist}^{+\infty}e^{-\ampNor u}I(\ampNor u)du\right)^2} &\\
     &= \frac{\frac{\Gamma(\ratioSucTeach)^2}{2\ampNor}e^{-\ampNor \dbdist}I(\ampNor \dbdist) - \frac{1}{\ampNor} I(0)\int_{\ampNor \dbdist}^{+\infty}e^{-v}I(v)dv}{\left(\frac{1}{\ampNor}\int_{\ampNor \dbdist}^{+\infty}e^{-v}I(v)dv\right)^2} &\text{(by the substitutions $v \!=\! \ampNor u$)}\\
     &= \ampNor \frac{\frac{\Gamma(\ratioSucTeach)^2}{2}e^{-\ampNor \dbdist}I(\ampNor \dbdist) - I(0)\int_{\ampNor \dbdist}^{+\infty}e^{-v}I(v)dv}{\left(\int_{\ampNor \dbdist}^{+\infty}e^{-v}I(v)dv\right)^2}.
 \end{align*}
\end{proof}

\begin{customthm}{2}
    The aggregation mechanism $\RandMech$ is $(\epsilon, 0)$-differentially private, with
    \begin{align*}
            \epsilon = \log\left(1 + 2\frac{\int_0^{\ampNor}e^{-v}I(v)dv}{\int_{2\ampNor}^{+\infty}e^{-v}I(v)dv}\right).
    \end{align*}
    Moreover, if $\ratioSucTeach > \frac{1}{2}$, g is differentiable in $0$ and $\RandMech$ is $(\epsilon', 0)$-differentially private, with
    \begin{align*}
            \epsilon' = \min \left[\epsilon, \log\left(g(0) - g'(0)\right)\right].
    \end{align*}
\end{customthm}

\begin{proof}
    Thanks to Lemma~\ref{lemma:max_of_g}, we can use Propositions~\ref{prop:maj_g_gamma} and \ref{prop:maj_g_tau} to upper bound $g$, for $\dbdist=2$. We then just have to apply Lemma~\ref{lemma:mimic_dwork} to conclude.
\end{proof}

\begin{lemma}
    \label{lemma:maj_I}
    For all $v\in \mathbb{R}^*_+$, $I(v) \leq v^{\ratioSucTeach-1}\frac{\Gamma(\ratioSucTeach)}{2^{\ratioSucTeach}}$.
\end{lemma}

\begin{proof}
    Let $v\in \mathbb{R}^*_+$.
    \begin{align*}
        I(v) &= \int_{0}^{+\infty}\left(t+v\right)^{\ratioSucTeach-1}t^{\ratioSucTeach-1}e^{-2t} dt&\\
        &\leq v^{\ratioSucTeach-1} \int_{0}^{+\infty}t^{\ratioSucTeach-1}e^{-2t} dt &\text{(because $\ratioSucTeach-1\leq 0$)}\\
        &= v^{\ratioSucTeach-1} \int_{0}^{+\infty}\left(\frac{u}{2}\right)^{\ratioSucTeach-1}e^{-u} \frac{du}{2} &\text{(by the substitution $u = 2t$)}\\
        &= v^{\ratioSucTeach-1} \frac{\Gamma(\ratioSucTeach)}{2^{\ratioSucTeach}}&
    \end{align*}
\end{proof}

\begin{customprop}{3}
    For all $\ratioSucTeach\in (0, 1)$, $\underset{\ampNor\to 0}{\lim}\left[\log\left(1 + 2\frac{\int_0^{\ampNor}e^{-v}I(v)dv}{\int_{2\ampNor}^{+\infty}e^{-v}I(v)dv}\right)\right] = 0$.
\end{customprop}

\begin{proof}
    For all $v\in \mathbb{R}^*_+$, $e^{-v}I(v)>0$ thus, supposing $\ampNor\in (0, 1]$, $\int_{2\ampNor}^{+\infty}e^{-v}I(v)dv \geq \int_2^{+\infty}e^{-v}I(v)dv > 0$. Therefore, it suffices to prove that $\underset{\ampNor\to 0}{\lim}\left[\int_0^{\ampNor}e^{-v}I(v)dv\right] = 0$ to deduce the announced result.
    
    Applying Lemma~\ref{lemma:maj_I}, we get
    \begin{align*}
        \int_0^{\ampNor}e^{-v}I(v)dv &\leq \frac{\Gamma(\ratioSucTeach)}{2^{\ratioSucTeach}} \int_0^{\ampNor}e^{-v}v^{\ratioSucTeach-1}dv\\
        &\leq \frac{\Gamma(\ratioSucTeach)}{2^{\ratioSucTeach}} \int_0^{\ampNor}v^{\ratioSucTeach-1}dv\\
        &= \frac{\Gamma(\ratioSucTeach)}{2^{\ratioSucTeach}} \frac{\ampNor^{\ratioSucTeach}}{\ratioSucTeach}\\
    \end{align*}
    which gives $\underset{\ampNor\to 0}{\lim}\left[\int_0^{\ampNor}e^{-v}I(v)dv\right] = 0$.
\end{proof}

\begin{customprop}{2}
    For all $\ampNor\in \mathbb{R}_+^*$, $\underset{\ratioSucTeach\to 1}{\lim}\left[\log\left(g(0) - g'(0)\right)\right] = 2\ampNor$.
\end{customprop}

\begin{proof}
    We use the dominated convergence theorem to determine the limit of $\density$ and $\cumulDensity$ when $\ratioSucTeach$ approaches $1$.
    Let us suppose in the following that $\ratioSucTeach\in (\frac{3}{4}; 1)$.
    
    First of all, we determine the limit of $I$ and deduce the one of $\density$. Let $v\in \mathbb{R}_+$.
    
    For all $x\in (0; 1]$, $(x+v)^{\ratioSucTeach-1}x^{\ratioSucTeach-1}e^{-2x} \leq x^{2\ratioSucTeach-2}e^{-2x} \leq x^{-\frac{1}{2}}e^{-2x}$.
    As $x\mapsto x^{-\frac{1}{2}}e^{-2x}$ is integrable on $(0; 1]$, and, for all $x\in (0;1]$,\\
    $\underset{\ratioSucTeach\to 1}{\lim}\left[(x+v)^{\ratioSucTeach-1}x^{\ratioSucTeach-1}e^{-2x}\right] = e^{-2x}$, by the dominated convergence theorem we get that $\underset{\ratioSucTeach\to 1}{\lim}\left[\int_0^1(x+v)^{\ratioSucTeach-1}x^{\ratioSucTeach-1}e^{-2x}dx\right] = \int_0^1e^{-2x}dx$.
    
    Similarly, as, for all $x\in [1; +\infty)$, $(x+v)^{\ratioSucTeach-1}x^{\ratioSucTeach-1}e^{-2x} \leq e^{-2x}$ and\\
    $\underset{\ratioSucTeach\to 1}{\lim}\left[(x+v)^{\ratioSucTeach-1}x^{\ratioSucTeach-1}e^{-2x}\right] = e^{-2x}$, by the dominated convergence theorem,\\
    $\underset{\ratioSucTeach\to 1}{\lim}\left[\int_1^{+\infty}(x+v)^{\ratioSucTeach-1}x^{\ratioSucTeach-1}e^{-2x}dx\right] = \int_1^{+\infty}e^{-2x}dx$.
    
    From the two points above, we deduce that
    \begin{align*}
        \underset{\ratioSucTeach\to 1}{\lim}I(v) &= \underset{\ratioSucTeach\to 1}{\lim}\left[\int_0^1(x+v)^{\ratioSucTeach-1}x^{\ratioSucTeach-1}e^{-2x}dx + \int_1^{+\infty}(x+v)^{\ratioSucTeach-1}x^{\ratioSucTeach-1}e^{-2x}dx\right]\\
        &= \int_0^1e^{-2x}dx + \int_1^{+\infty}e^{-2x}dx\\
        &= \int_0^{+\infty}e^{-2x}dx\\
        &= \frac{1}{2}
    \end{align*}
    and, for any $u\in \mathbb{R}$, $\underset{\ratioSucTeach\to 1}{\lim}\density(u) = \underset{\ratioSucTeach\to 1}{\lim}\left[\frac{\ampNor}{\Gamma(\ratioSucTeach)^2} e^{-\ampNor \abs{u}}I(\ampNor \abs{u})\right] = \frac{1}{2}\ampNor e^{-\ampNor \abs{u}}$.
    
    Let us now determine the limit of $\cumulDensity$.
    
    Let $u_0\in [0; \frac{1}{\ampNor}]$ and $u_1\in [0; \frac{1}{\ampNor}]$ such that $u_0 < u_1$.
    According to Lemma~\ref{lemma:maj_I}, for all $u\in (u_0; u_1]$, $e^{-\ampNor u}I(\ampNor u)\leq e^{-\ampNor u}(\ampNor u)^{\ratioSucTeach-1} \frac{\Gamma(\ratioSucTeach)}{2^{\ratioSucTeach}} \leq e^{-\ampNor u}(\ampNor u)^{-\frac{1}{4}} \frac{\Gamma(\frac{3}{4})}{2^{\frac{3}{4}}}$ because $\ampNor u \leq 1$ and $\Gamma$ is decreasing on $(0; 1]$.
    Since $u\mapsto e^{-\ampNor u}(\ampNor u)^{-\frac{1}{4}} \frac{\Gamma(\frac{3}{4})}{2^{\frac{3}{4}}}$ is integrable on $(u_0; u_1]$ and, for all $u\in (u_0; u_1]$, $\underset{\ratioSucTeach\to 1}{\lim}\left[e^{-\ampNor u}I(\ampNor u)\right] = \frac{e^{-\ampNor u}}{2}$, by the dominated convergence theorem, $\underset{\ratioSucTeach\to 1}{\lim}\left[\int_{u_0}^{u_1} e^{-\ampNor u}I(\ampNor u)du\right] = \int_{u_0}^{u_1}\frac{e^{-\ampNor u}}{2}du$.
    
    Let $u_0\in [\frac{1}{\ampNor}; +\infty)$ and $u_1\in [\frac{1}{\ampNor}; +\infty) \cup \{+\infty\}$ such that $u_0 < u_1$.
    Similarly, as, for all $u\in [u_0; u_1)$, $e^{-\ampNor u}I(\ampNor u)\leq e^{-\ampNor u}(\ampNor u)^{\ratioSucTeach-1} \frac{\Gamma(\ratioSucTeach)}{2^{\ratioSucTeach}} \leq e^{-\ampNor u} \frac{\Gamma(\frac{3}{4})}{2^{\frac{3}{4}}}$.
    Since $u\mapsto e^{-\ampNor u} \frac{\Gamma(\frac{3}{4})}{2^{\frac{3}{4}}}$ is integrable on $[u_0; u_1)$ and, for all $u\in [u_0; u_1)$, $\underset{\ratioSucTeach\to 1}{\lim}\left[e^{-\ampNor u}I(\ampNor u)\right] = \frac{e^{-\ampNor u}}{2}$, by the dominated convergence theorem,\\
    $\underset{\ratioSucTeach\to 1}{\lim}\left[\int_{u_0}^{u_1} e^{-\ampNor u}I(\ampNor u)du\right] = \int_{u_0}^{u_1}\frac{e^{-\ampNor u}}{2}du$.
    
    We deduce that, whatever are the bounds $u_0\in [0; +\infty)$ and $u_1\in [0; +\infty) \cup \{+\infty\}$ with $u_0 < u_1$, $\underset{\ratioSucTeach\to 1}{\lim}\left[\int_{u_0}^{u_1} e^{-\ampNor u}I(\ampNor u)du\right] = \int_{u_0}^{u_1}\frac{e^{-\ampNor u}}{2}du$.
    By substitution, we also have $\underset{\ratioSucTeach\to 1}{\lim}\left[\int_{u_0}^{u_1} e^{\ampNor u}I(-\ampNor u)du\right] = \int_{u_0}^{u_1}\frac{e^{\ampNor u}}{2}du$ for any $u_0\in (-\infty; 0] \cup \{-\infty\}$ and $u_1\in (-\infty; 0]$ with $u_0 < u_1$.
    
    Finally, for any $u_0\in (-\infty; 0] \cup \{-\infty\}$ and $u_1\in [0; +\infty) \cup \{+\infty\}$ such that $u_0< u_1$, we have $\underset{\ratioSucTeach\to 1}{\lim}\left[\int_{u_0}^{u_1} e^{-\ampNor \abs{u}}I(\ampNor \abs{u})du\right] = \int_{u_0}^{u_1}\frac{e^{-\ampNor \abs{u}}}{2}du$.
    In particular, for all $z\in \mathbb{R}$,
    \begin{align*}
        \underset{\ratioSucTeach\to 1}{\lim}\cumulDensity(z) &= \underset{\ratioSucTeach\to 1}{\lim}(\normalInt) \times \int_{-\infty}^{z}\frac{e^{-\ampNor \abs{u}}}{2}du\\
        &= \ampNor \int_{-\infty}^{z}\frac{e^{-\ampNor \abs{u}}}{2}du\\
        &= 
        \begin{cases}
          \frac{1}{2} e^{\ampNor z} \text{ if } z < 0 \\
          1 - \frac{1}{2} e^{-\ampNor z} \text{ if } z \geq 0
        \end{cases}
    \end{align*}
    which is actually the expression of the Laplace cumulative distribution function.
    
    From what precedes we can conclude that, with $\dbdist=2$,
    \begin{align*}
        &\underset{\ratioSucTeach\to 1}{\lim}\left[g(0)-g'(0)\right]\\
        &\quad = \underset{\ratioSucTeach\to 1}{\lim}\left[\frac{1-\cumulDensity(0)}{1-\cumulDensity(2)}-\frac{(1-\cumulDensity(0))\density(2) - (1-\cumulDensity(2))\density(0)}{(1-\cumulDensity(2))^2}\right]\\
        &\quad = \frac{\frac{1}{2}}{\frac{1}{2}e^{-2\ampNor}}-\frac{1}{2}\frac{\frac{1}{2}\times \frac{1}{2}\ampNor e^{-2\ampNor} - \frac{1}{2}e^{-2\ampNor}\times \frac{1}{2}\ampNor}{(\frac{1}{2}e^{-2\ampNor})^2}\\
        &\quad = e^{2\ampNor}.
    \end{align*}
\end{proof}

\subsection{Influence of the HE layer on the DP guarantee per query}
The computation of the homomorphic argmax induces some perturbations on the noisy counts and, as such, could harm the DP guarantees that we just gave.
The three kinds of perturbations due to the HE layer are:
\begin{itemize}
    \item the addition of (Gaussian) noise at the time of TFHE encryption which is inherently probabilistic
    \item the addition of a constant value $A$ on the noisy counts to ensure that all the noisy counts are positive (with high probability) (see Section~\ref{sec:fhe_argmax})
    \item a possible mistake on the argmax if two noisy counts are too close (see Section 6 of the main paper).
\end{itemize}

While these perturbations can be seen as some postprocessing applied on the clear noisy histogram, they cannot be seen as a postprocessing on the clear noisy argmax on which we showed DP guarantees in Section~\ref{sec:dp_per_query}. Nevertheless, if we can prove that these perturbations consist of an addition of noise on the clear histogram, the upper bound on $\frac{\mathbb{P}[r\geq t]}{\mathbb{P}[r\geq t+2]}$, $r$ being the total noise (generalised Laplace noise and HE perturbations) applied to the histogram of the $n_k$'s, would still hold, leading to the same DP guarantees. The additions of Gaussian noise and constant $A$ at encryption have, by commutativity, the same effect as the addition of a sum of Gaussian noises and $nA$ after summation and they will anyway change the output of the homomorphic argmax with very low probability. However, some further work needs to be done in order to check whether the third kind of perturbation can be simulated as a noise addition on the histogram.

\subsection{Upper bound of the probability of a report noisy max mistake}
\label{sec:upperbound_proba_q}
In this subsection, we give an upper bound of the probability that $\RandMech$ outputs a wrong argmax because of the added noise following the generalised Laplace distribution.

\begin{lemma}
    \label{lemma:maj_int_exp_I}
    Let $u_0\in \mathbb{R}_+$.
    Let $q\in \left(\frac{1}{1-\ratioSucTeach}; +\infty\right)$ and $p := \frac{1}{1-\frac{1}{q}}$.
    
    We have
    \begin{align*}
        \int_{u_0}^{+\infty} e^{-\ampNor u}I(\ampNor u)du \leq \frac{\Gamma(\ratioSucTeach)}{2^\ratioSucTeach \ampNor}\frac{e^{-\ampNor u_0}}{p^{\frac{1}{p}}}\frac{(\ampNor u_0)^{\ratioSucTeach-1+\frac{1}{q}}}{[q(1-\ratioSucTeach)-1]^{\frac{1}{q}}}.
    \end{align*}
\end{lemma}

\begin{proof}
    Let $u_0\in \mathbb{R}_+$.
    Let $(p, q)\in \left(\mathbb{R}_+^*\right)^2$ such that $\frac{1}{p} + \frac{1}{q} = 1$ and $q > \frac{1}{1-\ratioSucTeach}$.
    \begin{align*}
        &\int_{u_0}^{+\infty} e^{-\ampNor u}I(\ampNor u)du&\\
        &\quad\leq \frac{\Gamma(\ratioSucTeach)}{2^\ratioSucTeach}\int_{u_0}^{+\infty} e^{-\ampNor u}(\ampNor u)^{\ratioSucTeach-1}du &\text{(according to Lemma~\ref{lemma:maj_I})}\\
        &\quad= \frac{\Gamma(\ratioSucTeach)}{2^\ratioSucTeach \ampNor}\int_{\ampNor u_0}^{+\infty} e^{-v}v^{\ratioSucTeach-1}dv &\text{(by the substitution $v = \ampNor u$)}
    \end{align*}
    By assumption, $q > \frac{1}{1-\ratioSucTeach}$ so, since $\ratioSucTeach < 1$, $q(\ratioSucTeach-1) < -1$ and then $v\in \mathbb{R}_+^* \mapsto v^{q(\ratioSucTeach-1)}$ is integrable in the neighbourhood of $+\infty$. Then we can apply Hölder's inequality in the following manner:
    \begin{align*}
        &\int_{u_0}^{+\infty} e^{-\ampNor u}I(\ampNor u)du&\\
        &\quad\leq \frac{\Gamma(\ratioSucTeach)}{2^\ratioSucTeach \ampNor}\left(\int_{\ampNor u_0}^{+\infty} e^{-pv}dv\right)^\frac{1}{p}\left(\int_{\ampNor u_0}^{+\infty} v^{q(\ratioSucTeach-1)}dv\right)^\frac{1}{q}&\\
        &\quad= \frac{\Gamma(\ratioSucTeach)}{2^\ratioSucTeach \ampNor} \times \left(\frac{e^{-p\ampNor u_0}}{p}\right)^\frac{1}{p} \times \left(\frac{-(\ampNor u_0)^{q(\ratioSucTeach-1)+1}}{(q(\ratioSucTeach-1)+1)}\right)^{\frac{1}{q}}&\\
        &\quad= \frac{\Gamma(\ratioSucTeach)}{2^\ratioSucTeach \ampNor} \times \frac{e^{-\ampNor u_0}}{p^{\frac{1}{p}}} \times \frac{(\ampNor u_0)^{\ratioSucTeach-1+\frac{1}{q}}}{[q(1-\ratioSucTeach)-1]^{\frac{1}{q}}}.
    \end{align*}
\end{proof}

\begin{lemma}
    \label{lemma:upperbound_proba_q}
    Let us consider a query $Q$.
    Let $k^*\in [K]$ be the unnoisy argmax (for all $k\in [K]$, $n_{k^*}\geq n_k$). For all $k\in [K]$, we define $\Delta_k := n_{k^*} - n_k \geq 0$. Then, for all $q\in (\frac{1}{1-\ratioSucTeach}; +\infty)$, calling $p:=\frac{1}{1-\frac{1}{q}}$,
    \begin{align*}
        \mathbb{P}[\RandMech(d, Q)\neq k^*] \leq \sum_{k\neq k^*}e^{-\ampNor \Delta_k}\left[\frac{1}{2} + \frac{1}{\ratioSucTeach 2^{4\ratioSucTeach-2+\frac{1}{q}} \Gamma(\ratioSucTeach)^2} \times \frac{(\ampNor \Delta_k)^{2\ratioSucTeach-1+\frac{1}{q}}}{p^{\frac{1}{p}}[q(1-\ratioSucTeach)-1]^{\frac{1}{q}}}\right].
    \end{align*}
\end{lemma}

\begin{proof}
    In the following, we will assume that $\Delta_k > 0$ and the upper bound for $\Delta_k=0$ is obtained by continuity.
    
    For any $k\in [K]$, let us denote $Y_k$ the random variable following the generalised Laplace distribution generated by the sum of the $\tau n$ individual noises.
    
    Let $k\in [K]$.
    \begin{align*}
        &\mathbb{P}(n_k+Y_k \geq n_{k^*} + Y_{k^*})&\\
        &\quad= \mathbb{P}(Y_{k^*} \leq Y_k - \Delta_k)&\\
        &\quad= \int_{-\infty}^{+\infty} \density(t) \cumulDensity(t-\Delta_k)dt&\\
        &\quad= \int_{-\infty}^0 \density(t) \cumulDensity(t-\Delta_k)dt + \int_0^{\Delta_k} \density(t) \cumulDensity(t-\Delta_k)dt&\\
        &\quad\qquad + \int_{\Delta_k}^{+\infty} \density(t) \cumulDensity(t-\Delta_k)dt &\numberthis \label{eq:P_split_integrals}
    \end{align*}
    
    We will now upper bound each one of the three above integrals separately.
    The two extreme integrals can be nicely bounded by decreasing exponentials in $\Delta_k$:
    \begin{align*}
        &\int_{\Delta_k}^{+\infty} \density(t) \cumulDensity(t-\Delta_k)dt&\\
        &\quad= \int_{0}^{+\infty} \density(v+\Delta_k) \cumulDensity(v)dv &\text{(by the substitution $v = t - \Delta_k$)}\\
        &\quad= \normalInt \int_{0}^{+\infty} e^{-\ampNor\abs{v+\Delta_k}}I(\ampNor \abs{v+\Delta_k}) \cumulDensity(v)dv&\\
        &\quad= \normalInt \int_{0}^{+\infty} e^{-\ampNor(v+\Delta_k)}I(\ampNor(v+\Delta_k)) \cumulDensity(v)dv&\\
        &\quad= \normalInt e^{-\ampNor\Delta_k}\int_{0}^{+\infty} e^{-\ampNor v}I(\ampNor(v+\Delta_k)) \cumulDensity(v)dv&\\
        &\quad \leq \normalInt e^{-\ampNor\Delta_k}\int_{0}^{+\infty} e^{-\ampNor v}I(\ampNor v) \cumulDensity(v)dv &\text{(because $I$ is decreasing)}\\
        &\quad= \normalInt e^{-\ampNor\Delta_k}\int_{0}^{+\infty} e^{-\ampNor \abs{v}}I(\ampNor \abs{v}) \cumulDensity(v)dv&\\
        &\quad= e^{-\ampNor\Delta_k}\int_{0}^{+\infty} \density(v) \cumulDensity(v)dv&\\
        &\quad= e^{-\ampNor \Delta_k}\times \frac{\underset{+\infty}{\lim}F^2 - F(0)^2}{2}&\\
        &\quad= e^{-\ampNor \Delta_k}\times \frac{1 - \frac{1}{4}}{2}&\\
        &\quad= \frac{3}{8}e^{-\ampNor \Delta_k} \numberthis \label{eq:first_integral}
    \end{align*}
    and
    
    \begin{align*}
        &\int_{-\infty}^0 \density(t) \cumulDensity(t-\Delta_k)dt&\\
        &\quad= \normalInt \int_{-\infty}^0 \density(t) \int_{-\infty}^{t-\Delta_k}e^{-\ampNor \abs{u}}I(\ampNor \abs{u})dudt&\\
        &\quad= \normalInt \int_{-\infty}^0 \density(t) \int_{-\infty}^{t-\Delta_k}e^{\ampNor u}I(-\ampNor u)dudt&\\
        &\quad= \normalInt \int_{-\infty}^0 \density(t) \int_{-\infty}^{t}e^{\ampNor (v-\Delta_k)}I(\ampNor(\Delta_k-v))dudt\\
        &\quad \qquad\qquad\qquad\qquad \text{(by the substitution $v = u + \Delta_k$)} &\\
        &\quad= \normalInt e^{-\ampNor \Delta_k} \int_{-\infty}^0 \density(t) \int_{-\infty}^{t}e^{\ampNor v}I(\ampNor(\Delta_k-v))dudt&\\
        &\quad\leq \normalInt e^{-\ampNor \Delta_k} \int_{-\infty}^0 \density(t) \int_{-\infty}^{t}e^{\ampNor v}I(-\ampNor v)dudt &\text{(because $I$ is decreasing)}\\
        &\quad= \normalInt e^{-\ampNor \Delta_k} \int_{-\infty}^0 \density(t) \int_{-\infty}^{t}e^{-\ampNor \abs{v}}I(\ampNor \abs{v})dudt&\\
        &\quad= e^{-\ampNor \Delta_k} \int_{-\infty}^0 \density(t) \cumulDensity(t)dt&\\
        &\quad= e^{-\ampNor \Delta_k} \times \frac{\cumulDensity(0)^2 - \underset{-\infty}{\lim}F^2}{2}&\\
        &\quad= \frac{1}{8}e^{-\ampNor \Delta_k}. \numberthis \label{eq:second_integral}
    \end{align*}
    
    As for the middle integral, we have
    \begin{align*}
        &\int_0^{\Delta_k} \density(t) \cumulDensity(t-\Delta_k)dt&\\
        &\quad= \normalInt \int_0^{\Delta_k} \density(t) \int_{-\infty}^{t-\Delta_k} e^{-\ampNor \abs{u}} I(\ampNor \abs{u})dudt&\\
        &\quad= \normalInt \int_0^{\Delta_k} \density(t) \int_{\Delta_k-t}^{+\infty} e^{-\ampNor \abs{v}} I(\ampNor \abs{v})dvdt &\text{(by the substitution $v = -u$)}\\
        &\quad= \normalInt \int_0^{\Delta_k} \density(t) \int_{\Delta_k-t}^{+\infty} e^{-\ampNor v} I(\ampNor v)dvdt.
    \end{align*}
    
    Since, for all $t\in [0; \Delta_k]$, $0\leq \Delta_k-t$, we can apply Lemma~\ref{lemma:maj_int_exp_I}.
    Let $q\in \left(\frac{1}{1-\ratioSucTeach}; +\infty\right)$ and $p = \frac{1}{1-\frac{1}{q}}$.
    We have, for all $t\in (0; \Delta_k)$, $\int_{\Delta_k-t}^{+\infty} e^{-\ampNor v} I(\ampNor v)dv \leq \frac{\Gamma(\ratioSucTeach)}{2^\ratioSucTeach \ampNor} \times \frac{1}{p^{\frac{1}{p}}[q(1-\ratioSucTeach)-1]^{\frac{1}{q}}} \times e^{-\ampNor (\Delta_k-t)} [\ampNor (\Delta_k-t)]^{\ratioSucTeach-1+\frac{1}{q}}$. Since $\ratioSucTeach-1+\frac{1}{q} > -1$, $t\mapsto [\ampNor (\Delta_k-t)]^{\ratioSucTeach-1+\frac{1}{q}}$ is integrable on a neighbourhood of $(\Delta_k)^-$ and then, since $t\mapsto \density(t) e^{-\ampNor (\Delta_k-t)}$ is bounded on a neighbourhood of $\Delta_k$, $t\mapsto \density(t) e^{-\ampNor (\Delta_k-t)}[\ampNor (\Delta_k-t)]^{\ratioSucTeach-1+\frac{1}{q}}$ is integrable on a neighbourhood of $(\Delta_k)^-$.
    
    Thus, we can write
    \begin{align*}
        &\int_0^{\Delta_k} \density(t) \cumulDensity(t-\Delta_k)dt&\\
        &\quad\leq \normalInt \frac{\Gamma(\ratioSucTeach)}{2^\ratioSucTeach \ampNor} \times \frac{1}{p^{\frac{1}{p}}[q(1-\ratioSucTeach)-1]^{\frac{1}{q}}} \times \int_0^{\Delta_k} \density(t) e^{-\ampNor (\Delta_k-t)} [\ampNor (\Delta_k-t)]^{\ratioSucTeach-1+\frac{1}{q}}dt&\\
        &\quad= \normalInt^2 \frac{\Gamma(\ratioSucTeach)}{2^\ratioSucTeach \ampNor} \times \frac{1}{p^{\frac{1}{p}}[q(1-\ratioSucTeach)-1]^{\frac{1}{q}}}\\
        &\quad \qquad\qquad\qquad \times \int_0^{\Delta_k} e^{-\ampNor \abs{t}} I(\ampNor \abs{t}) e^{-\ampNor (\Delta_k-t)} [\ampNor (\Delta_k-t)]^{\ratioSucTeach-1+\frac{1}{q}}dt&\\
        &\quad= \frac{\ampNor}{2^\ratioSucTeach \Gamma(\ratioSucTeach)^3} \times \frac{1}{p^{\frac{1}{p}}[q(1-\ratioSucTeach)-1]^{\frac{1}{q}}}\\
        &\quad \qquad\qquad\qquad \times \int_0^{\Delta_k} e^{-\ampNor t} I(\ampNor t) e^{-\ampNor (\Delta_k-t)} [\ampNor (\Delta_k-t)]^{\ratioSucTeach-1+\frac{1}{q}}dt&\\
        &\quad=  \frac{e^{-\ampNor \Delta_k}}{2^\ratioSucTeach \Gamma(\ratioSucTeach)^3} \times \frac{\ampNor}{p^{\frac{1}{p}}[q(1-\ratioSucTeach)-1]^{\frac{1}{q}}} \times \int_0^{\Delta_k} I(\ampNor t)[\ampNor (\Delta_k-t)]^{\ratioSucTeach-1+\frac{1}{q}}dt.
    \end{align*}
    
    $t\mapsto (\ampNor t)^{\ratioSucTeach-1}$ is integrable on a neighbourhood of $0^+$ because $\ratioSucTeach - 1 > -1$. Therefore, $t\mapsto (\ampNor t)^{\ratioSucTeach-1} \frac{\Gamma(\ratioSucTeach)}{2^\ratioSucTeach}[\ampNor (\Delta_k-t)]^{\ratioSucTeach-1+\frac{1}{q}}$ is integrable on $(0; \Delta_k)$ so we can apply Lemma~\ref{lemma:maj_I}:
    \begin{align*}
        &\int_0^{\Delta_k} \density(t) \cumulDensity(t-\Delta_k)dt&\\
        &\quad\leq \frac{e^{-\ampNor \Delta_k}}{2^\ratioSucTeach \Gamma(\ratioSucTeach)^3} \times \frac{\ampNor}{p^{\frac{1}{p}}[q(1-\ratioSucTeach)-1]^{\frac{1}{q}}} \times \int_0^{\Delta_k} (\ampNor t)^{\ratioSucTeach-1} \frac{\Gamma(\ratioSucTeach)}{2^\ratioSucTeach}[\ampNor (\Delta_k-t)]^{\ratioSucTeach-1+\frac{1}{q}}dt &\\
        &\quad= \frac{e^{-\ampNor \Delta_k}}{2^{2\ratioSucTeach} \Gamma(\ratioSucTeach)^2} \times \frac{\ampNor \Delta_k}{p^{\frac{1}{p}}[q(1-\ratioSucTeach)-1]^{\frac{1}{q}}} \times \int_0^1 (\ampNor \Delta_k u)^{\ratioSucTeach-1} [\ampNor (\Delta_k-\Delta_k u)]^{\ratioSucTeach-1+\frac{1}{q}}du\\
        &\quad \qquad\qquad\qquad\qquad\qquad\qquad\qquad\qquad\qquad\qquad \text{(by the substitution $u = \frac{t}{\Delta_k}$)}\\
        &\quad= \frac{e^{-\ampNor \Delta_k}}{2^{2\ratioSucTeach} \Gamma(\ratioSucTeach)^2} \times \frac{(\ampNor \Delta_k)^{2\ratioSucTeach-1+\frac{1}{q}}}{p^{\frac{1}{p}}[q(1-\ratioSucTeach)-1]^{\frac{1}{q}}} \times \int_0^1 u^{\ratioSucTeach-1} (1-u)^{\ratioSucTeach-1+\frac{1}{q}}du.
    \end{align*}
    
    Note that
    \begin{align*}
        &\int_0^1 u^{\ratioSucTeach-1} (1-u)^{\ratioSucTeach-1+\frac{1}{q}}du&\\
        &\quad= \int_0^{\frac{1}{2}} u^{\ratioSucTeach-1} (1-u)^{\ratioSucTeach-1+\frac{1}{q}}du + \int_{\frac{1}{2}}^1 u^{\ratioSucTeach-1} (1-u)^{\ratioSucTeach-1+\frac{1}{q}}du&\\
        &\quad\leq \int_0^{\frac{1}{2}} u^{\ratioSucTeach-1} \frac{1}{2^{\ratioSucTeach-1+\frac{1}{q}}}du + \int_{\frac{1}{2}}^1 \frac{1}{2^{\ratioSucTeach-1}} (1-u)^{\ratioSucTeach-1+\frac{1}{q}}du\\
        &\quad \qquad\qquad\qquad\qquad\qquad\qquad\qquad\qquad \text{(because $\ratioSucTeach-1+\frac{1}{q} < 0$ and $\ratioSucTeach-1 < 0$)}\\
        &\quad= \frac{1}{2^{\ratioSucTeach-1+\frac{1}{q}}}\int_0^{\frac{1}{2}} u^{\ratioSucTeach-1} du +  \frac{1}{2^{\ratioSucTeach-1}} \int_0^{\frac{1}{2}} v^{\ratioSucTeach-1+\frac{1}{q}}dv\\
        &\quad \qquad\qquad\qquad\qquad\qquad\qquad\qquad\qquad\qquad\quad\,\, \text{(by the substitution $v = 1-u$)}\\
        &\quad= \frac{1}{2^{\ratioSucTeach-1+\frac{1}{q}}} \times \frac{1}{\ratioSucTeach 2^{\ratioSucTeach}} + \frac{1}{2^{\ratioSucTeach-1}} \times \frac{1}{(\ratioSucTeach+\frac{1}{q})2^{\ratioSucTeach+\frac{1}{q}}}&\\
        &\quad= \frac{1}{2^{2\ratioSucTeach-1+\frac{1}{q}}} \left(\frac{1}{\ratioSucTeach} + \frac{1}{\ratioSucTeach+\frac{1}{q}}\right)&\\
        &\quad\leq \frac{1}{\ratioSucTeach 2^{2\ratioSucTeach-2+\frac{1}{q}}}.
    \end{align*}
    
    Therefore
    \begin{align}
        \label{eq:third_integral}
        \int_0^{\Delta_k} \density(t) \cumulDensity(t-\Delta_k)dt \leq \frac{e^{-\ampNor \Delta_k}}{\ratioSucTeach 2^{4\ratioSucTeach-2+\frac{1}{q}} \Gamma(\ratioSucTeach)^2} \times \frac{(\ampNor \Delta_k)^{2\ratioSucTeach-1+\frac{1}{q}}}{p^{\frac{1}{p}}[q(1-\ratioSucTeach)-1]^{\frac{1}{q}}}.
    \end{align}
    
    Using \ref{eq:P_split_integrals}, \ref{eq:first_integral}, \ref{eq:second_integral} and \ref{eq:third_integral}, we get
    \begin{align*}
        \mathbb{P}(n_k+Y_k \geq n_{k^*} + Y_{k^*}) \leq e^{-\ampNor \Delta_k}\left[\frac{1}{2} + \frac{1}{\ratioSucTeach 2^{4\ratioSucTeach-2+\frac{1}{q}} \Gamma(\ratioSucTeach)^2} \times \frac{(\ampNor \Delta_k)^{2\ratioSucTeach-1+\frac{1}{q}}}{p^{\frac{1}{p}}[q(1-\ratioSucTeach)-1]^{\frac{1}{q}}}\right].
    \end{align*}
    
    The overall upper bound for $\mathbb{P}[\RandMech(d;Q)\neq k^*]$ is obtained using the fact that the event $(\RandMech(d;Q)\neq k^*)$ is the union of the events $(n_k+Y_k \geq n_{k^*} + Y_{k^*})$, for $k\in [K]\setminus \{k^*\}$, and then $\mathbb{P}[\RandMech(d;Q)\neq k^*] \leq \sum_{k\neq k^*} \mathbb{P}(n_k+Y_k \geq n_{k^*} + Y_{k^*})$.
\end{proof}

\begin{customprop}{4}
    If $\ratioSucTeach \in (\frac{1}{2}; 1)$,
    \begin{align*}
        \mathbb{P}[\RandMech(d;Q)\neq k^*] \leq \sum_{k\neq k^*} e^{-\ampNor \Delta_k} \left[\frac{1}{2} + \frac{(\ampNor \Delta_k)^{2\ratioSucTeach-1}}{\ratioSucTeach 2^{4\ratioSucTeach-2} \Gamma(\ratioSucTeach)^2}\right].
    \end{align*}
    If $\ratioSucTeach \in (0; \frac{1}{2}]$,
    \begin{align*}
        \mathbb{P}[\RandMech(d;Q)\neq k^*] \leq \sum_{k\neq k^*} e^{-\ampNor \Delta_k}\left[\frac{1}{2} + \frac{(\ampNor \Delta_k)^{\frac{\ratioSucTeach}{2}}}{\ratioSucTeach 2^{\frac{5}{2}\ratioSucTeach-1} \Gamma(\ratioSucTeach)^2} \times \left(\frac{3}{2}\ratioSucTeach\right)^{\frac{3}{2}\ratioSucTeach}\left(\frac{2}{\ratioSucTeach}-3\right)^{1-\frac{3}{2}\ratioSucTeach}\right].
    \end{align*}
\end{customprop}

\begin{proof}
    Let us distinct two cases according to the value of $\ratioSucTeach$.
    
    \textbf{\underline{First case}:} $\tau > \frac{1}{2}$
    
    Taking the limit when $q$ approaches $+\infty$ in~\ref{eq:third_integral} (which actually amounts to substitute $v^{\ratioSucTeach-1}$ by its upper bound $(\ampNor u_0)^{\ratioSucTeach-1}$ in the integral $\int_ {\ampNor u_0}^{+\infty}e^{-v}v^{\ratioSucTeach-1}dv$ of the proof of Lemma~\ref{lemma:maj_int_exp_I}, without needing Hölder's inequality), we get
    \begin{align*}
        \mathbb{P}[\RandMech(d;Q)\neq k^*] \leq \sum_{k\neq k^*} e^{-\ampNor \Delta_k} \left[\frac{1}{2} + \frac{(\ampNor \Delta_k)^{2\ratioSucTeach-1}}{\ratioSucTeach 2^{4\ratioSucTeach-2} \Gamma(\ratioSucTeach)^2}\right].
    \end{align*}
    
    \textbf{\underline{Second case}:} $\tau \leq \frac{1}{2}$
    
    By convention, if $\tau = \frac{1}{2}$, we have $\frac{1}{1-2\ratioSucTeach} = +\infty$.
    
    We take $q < \frac{1}{1-2\ratioSucTeach}$ (it is possible since $\frac{1}{1-2\ratioSucTeach} > \frac{1}{1-\ratioSucTeach}$) and write $q = \frac{1}{1-2\ratioSucTeach + \epsilon}$, with $0 < \epsilon < \ratioSucTeach$. Then, $\frac{1}{p} = 1-\frac{1}{q} = 2\ratioSucTeach-\epsilon$ and we get
    \begin{align*}
        \mathbb{P}[\RandMech(d;Q)\neq k^*]\leq \sum_{k\neq k^*} e^{-\ampNor \Delta_k}\left[\frac{1}{2} + \frac{(2\ratioSucTeach-\epsilon)^{2\ratioSucTeach-\epsilon}}{\ratioSucTeach 2^{2\ratioSucTeach-1+\epsilon} \Gamma(\ratioSucTeach)^2} \times \left(\frac{1-2\ratioSucTeach+\epsilon}{\ratioSucTeach-\epsilon}\right)^{1-2\ratioSucTeach+\epsilon} \times (\ampNor \Delta_k)^{\epsilon}\right].
    \end{align*}
    
    For example, with $\epsilon=\frac{\ratioSucTeach}{2}$ (i.e. $q=\frac{1}{1-\frac{3}{2}\ratioSucTeach}$), we have
    \begin{align*}
        \mathbb{P}[\RandMech(d;Q)\neq k^*]\leq \sum_{k\neq k^*} e^{-\ampNor \Delta_k}\left[\frac{1}{2} + \frac{1}{\ratioSucTeach 2^{\frac{5}{2}\ratioSucTeach-1} \Gamma(\ratioSucTeach)^2} \times \left(\frac{3}{2}\ratioSucTeach\right)^{\frac{3}{2}\ratioSucTeach}\left(\frac{2}{\ratioSucTeach}-3\right)^{1-\frac{3}{2}\ratioSucTeach} \times (\ampNor \Delta_k)^{\frac{\ratioSucTeach}{2}}\right].
    \end{align*}
\end{proof}

Note that, whatever is the value of $\ratioSucTeach \in (0; 1)$, our upper bound of $\mathbb{P}(n_k+Y_k \geq n_{k^*} + Y_{k^*})$ tends to $0$ when $\Delta_k$ approaches $+\infty$ which follows the intuition that $\mathbb{P}(n_k+Y_k \geq n_{k^*} + Y_{k^*})$ tends to $0$ when the true argmax $k^*$ has a much higher count than $k$. The upper bound tends to $\frac{1}{2}$ when $\Delta_k$ approaches $0$, which is consistent with the actual value of the probability $\mathbb{P}(n_k+Y_k \geq n_{k^*} + Y_{k^*})$ when the counts $n_{k^*}$ and $n_k$ are equal.

Similarly, the upper bound tends to $0$ when $\ampNor$ tends to $+\infty$ and to $\frac{1}{2}$ when $\ampNor$ approaches $0$. These are the expected values of the probability $\mathbb{P}(n_k+Y_k \geq n_{k^*} + Y_{k^*})$ when there is no noise or an infinitely wide noise respectively.

Finally, let us remark that we recover the upper bound $\mathbb{P}[\RandMech(d;Q)\neq k^*] \leq \sum_{k\neq k^*} \frac{2 + \ampNor \Delta_k}{4e^{\ampNor \Delta_k}}$ from~\cite{papernot2016semisupervised} (obtained with a centralised Laplace noise) when we consider the limit when $\ratioSucTeach$ tends to $1$.

\paragraph{Remark.} The data-dependent bound $\alpha_{\RandMech}(l;\aux,d,d') \!\leq\! \log\! \left(\!(1 \!-\! q)\!\left(\frac{1 -q}{1 -e^{\epsilon}q}\right)^l \!+\! qe^{\epsilon l}\right)$ from Theorem~\ref{th:AlgoMomentAccountance} is non-monotonic in $\ampNor$. This may appear counter-intuitive since a smaller noise (greater $\ampNor$) usually gives worse privacy guarantees and, as one would expect, a bigger moments accountant. Nevertheless, a smaller noise means that the probability of outputting the true (unnoisy) argmax is closer to $1$, which may lower the moments accountant. Indeed, two adjacent databases will both output the true argmax with high probability, giving less chance to an adversary to distinguish them. This non-monotonicity of the data-dependent bound induces the non-monotonicity of the overall privacy cost $\epsilon$. This is illustrated in Figure 2 of the paper on which we can see, however, that choosing a small $\ampNor$ still gives better guarantees.

\section{FHE argmax implementation details}
\label{sec:fhe_argmax}
We implemented the FHE argmax algorithm using the C++ TFHE library \cite{TFHE-lib}. Table \ref{table:fhe_parameters} presents all of the parameters needed to reproduce our results and build a fully homomorphic argmax scheme using the TFHE library. The first two lines present our values for the standard TFHE parameters: the first line for initial ciphertext encryption; the second line for the two bootstrapping keys we use. Given the parameters that we use here, we achieve a security parameter of $110$. We base the security of our scheme on the \texttt{lwe-estimator}\footnote
{
\url{https://bitbucket.org/malb/lwe-estimator/raw/HEAD/estimator.py}
}
script. The estimator is based on the work presented in \cite{APS15} and is consistently kept up to date.

\begin{table}[ht]
    \centering
    \caption{Parameters for our implementation. The top line presents the overall security ($\lambda$), and the parameters for the initial encryption: $\sigma$ is the Gaussian noise parameter and $N$ is the size of polynomials. In the TFHE encryption scheme, there is a parameter $k$ (different from the one used in this paper) which, in our case, is always equal to $1$. The second line presents the parameters needed to create the two bootstrapping keys we are using. For these two lines, we used the notations from \cite{embed2019} and \cite{ChillottiGGI16}. The third line presents parameters specific to our implementation given the specificities of the data to process. $\amplitude$ is the value to add to the ciphertexts before subtracting $n_k+Y_k - n_{k'} - Y_{k'}$ as per the notations in Section 4.3 of the paper. $\initModulus$ is the modulus with which the values are rescaled at encryption time to obtain values in $[0,1]$ and to allow for a correct result of the $\theta$ computation. $\bootModulus^{(1)}$ is the output modulus of the first bootstrapping operation creating the $\theta$ values. $\bootModulus^{(2)}$ is the output modulus of the second and final bootstrapping operation. \label{table:fhe_parameters}}
    \begin{tabular}{c c}
            $N$ & $\sigma$ \\
        \hline
            $1024$ & $1$e$-9$ \\
    \end{tabular}
    \\
    \begin{tabular}{c c c c }
            $N_b$ & $\sigma_b$ & $B_g$ & $\ell$ \\
        \hline
            $1024$ & $1e$-9 & $64$ & $6$ \\ 
    \end{tabular}
    \\
    \begin{tabular}{c c c c}
            $\amplitude$ & $\initModulus$ & $\bootModulus^{(1)}$ & $\bootModulus^{(2)}$ \\
        \hline
            $900$ & $4102$ & $36$ & $4$ \\ 
    \end{tabular}
\end{table}

The third line presents parameters that are specific to our implementation. Because of the use of Gamma distributions, the values sent by the teachers can be negative. This can be an important issue: if a value is negative, then it will be interpreted in the ciphertext space as a very high positive value and the resulting argmax will be wrong. Therefore, after summing the ciphertexts from the teachers, we add a constant value (we can add a clear value to a ciphertext value) $\amplitude$ to ensure that the $n_k+Y_k+\amplitude$ are all positive before subtraction. We evaluated that, given the parameters of the Gamma distributions used, choosing $\amplitude = 900$ gives us less than a $2^{-64}$ probability of failure: with $Y_k$ following a Laplace distribution (as seen in Section 4 of the paper), then we have $\mathbb{P}(Y_k < -A) < 2^{-64}$. The $\initModulus$ variable corresponds to the value by which we rescale the cleartexts before encryption. Indeed, the cleartext and ciphertext spaces of the TFHE encryption scheme are both $\mathbb{T} = ([0,1], +)$. Additionally, for a correct $\theta$ computation, we need to have $|\frac{n_k+Y_k - n_{k'} - Y_{k'}}{\initModulus}| < \frac{1}{2}$, which is true if, for all $k\in [K]$, $\frac{n_k+Y_k+\amplitude}{\initModulus} \in [0,\frac{1}{2})$. Since $\mathbb{P}(Y_k \geq A) < 2^{-64}$ by symmetry, $\initModulus = 2(n+2\amplitude) = 4100$ (with $n$ the number of teachers) is sufficient to have $|\frac{n_k+Y_k - n_{k'} - Y_{k'}}{\initModulus}| < \frac{1}{2}$ with high probability. $\bootModulus^{(1)}$ is the output modulus of the first bootstrapping operation. It needs to be chosen so that we have $\Theta_k > \frac{1}{2}$ for one and only one $k$. That $k$ will then be considered the argmax. $\bootModulus^{(2)}$ is the modulus for the final bootstrapping operation.

\section{Detailed experimental settings}
In this section, we provide the reader with additional details regarding experimental settings. In order to reproduce experimental results, all necessary source codes are available on \url{https://github.com/Arnaud-GS/SPEED}. 

\subsection{Experimental settings for MNIST}
Following PATE experimental conditions, we built our framework based on the code repositories\footnote{https://github.com/tensorflow/privacy/tree/master/research/pate\_2017} accompanying~\cite{papernot2016semisupervised}. The teacher models are based on two convolutional layers with max-pooling and one fully connected layer with ReLUs. Code modifications have been performed on the initial repository, and are available on \url{https://github.com/Arnaud-GS/SPEED}. The execution environment consists in Python 3 and Tensorflow 1.15.0. The batch size, learning rate and max steps parameters have been respectively set to 128, 0.01 and 5000. As stated in~\cite{papernot2016semisupervised}, this yields an aggregate test-error rate of 93\%. A semi-supervised technique proposed in~\cite{salimans2016improved} has been used\footnote{https://github.com/openai/improved-gan}, in an execution environment consisting of Python 3 and Theano 0.7. Besides modifications available on \url{https://github.com/Arnaud-GS/SPEED}, the learning rate and number of epochs have been set to 0.001 and 500 respectively.

\subsection{Experimental settings for SVHN}
For SVHN, two additional layers have been added to the teacher models which were learned using a node with 8 NVIDIA v100. The batch size, learning rate and max steps parameters have been respectively set to 64, 0.08 and 2000. The student model also uses the improved GAN semi-supervised model, relying on Python 3 and Theano 0.8.2. The learning rate and number of epochs have been set to 0.0003 and 600 respectively.

\end{document}